\documentclass[twocolumn]{autart}     

\usepackage{amssymb,mathrsfs,color}
\usepackage{array}
\usepackage{graphicx,color}
\usepackage{latexsym}
\usepackage{amsmath}

\usepackage{eclbkbox}

\usepackage{cite}

\usepackage{boxedminipage}
\usepackage{fancybox}
\usepackage{ascmac}
\usepackage{url}

\newtheorem{theorem}         {Theorem}[section]
\newtheorem{proof}         {Proof}

\newtheorem{proposition}            {Proposition}[section]
\newtheorem{definition}            {Definition}[section]
\newtheorem{assumption}            {Assumption}[section]

\begin{document}

\setlength{\abovedisplayskip}{7pt}  
\setlength{\belowdisplayskip}{7pt}  

\begin{frontmatter}

\title{
Retrofit Control with Approximate Environment Modeling
\vspace{-8mm}
\thanksref{footnoteinfo}}  

\thanks[footnoteinfo]{
This work was supported by JST CREST Grant Number JP-MJCR15K1, Japan.
Corresponding author: T.~Ishizaki, Tel. \&\ Fax: +81-3-5734-2646. 
}

\author[TIT]{Takayuki Ishizaki$^{*}$}\ead{ishizaki@sc.e.titech.ac.jp},    
\author[TIT]{Takahiro Kawaguchi}\ead{kawaguchi@cyb.sc.e.titech.ac.jp},
\author[TIT]{Hampei Sasahara}\ead{sasahara.h@cyb.sc.e.titech.ac.jp},
and
\author[TIT]{Jun-ichi Imura}\ead{imura@sc.e.titech.ac.jp}   

\address[TIT]{Tokyo Institute of Technology; 2-12-1, Ookayama, Meguro, Tokyo, 152-8552, Japan.}


\begin{keyword}\vspace{-6pt}                            
Retrofit control, Approximate modeling, Modular design, Network systems, Decentralized control.  
\end{keyword}                              

\begin{abstract}
In this paper, we develop a retrofit control method with approximate environment modeling.
Retrofit control is a modular control approach for a general stable network system whose subsystems are supposed to be managed by their corresponding subsystem operators.
From the standpoint of a single subsystem operator who performs the design of a retrofit controller, the  subsystems managed by all other operators can be regarded as an environment, the complete system model of which is assumed not to be available.
The proposed retrofit control with approximate environment modeling has an advantage that the stability of the resultant control system is robustly assured regardless of not only the stability of approximate environment models, but also the magnitude of modeling errors, provided that the network system before implementing retrofit control is originally stable.
This robustness property is practically significant to incorporate existing identification methods of unknown environments, because the accuracy of identified models may neither be reliable nor assurable in reality.
Furthermore, we conduct a control performance analysis to show that the resultant performance can be regulated by adjusting the accuracy of approximate environment modeling.
The efficiency of the proposed retrofit control is shown by numerical experiments on a network of second-order oscillators.\vspace{-6pt}
\end{abstract}\vspace{-15pt}

\end{frontmatter}


\section{Introduction}

A module is one of semi-independent parts or subsystems in an integrated system of components.
For example, in software development, a unit of programs that can be handled by an individual developer is called a software module \cite{parnas1972criteria}.
As pointed out in a broad range of literature \cite{huang1998modularity,schilling2000toward,baldwin2000design}, increasing the modularity in design is the key to developing large-scale complex systems with flexibility to meet heterogeneous demands.
Such ``modular design" enables multiple entities or subsystem operators to individually develop, modify, and replace respective modules or subsystems, serving for significant reduction of efforts to adjust and coordinate a family of integrated components.
This is a strong advantage as compared to ``integral design," where each component has strong interdependence among others \cite{ulrich2003role,ulrich2003product}.

For dynamical network systems, a modular design method of decentralized controllers has been introduced in the context of retrofit control \cite{ishizaki2018retrofit,sadamoto2018retrofit,inoue2018parameterization,sasahara2018parameterization}.
The retrofit control can be applied to a general stable network system whose subsystems are supposed to be managed by their corresponding subsystem operators.
From the standpoint of a single subsystem operator who performs the design of a retrofit controller, the subsystems managed by all other operators can be regarded as an \textit{environment}, the system model of which is assumed not to be available.
This reflects a practical situation where subsystem models, control policies, and demands of the other subsystem operators may not be public and stationary.

Most existing decentralized and distributed control methods, such as in \cite{siljak1972stability,rotkowitz2006characterization,d2003distributed,langbort2004distributed,vsiljak2005control,rantzer2015scalable}, can be classified as an integral design approach of structured controllers, where a single authority with availability of an entire system model is premised for simultaneous design of all subcontrollers constituting a decentralized or distributed controller.
In contrast, the retrofit control is classified as a modular design approach, where multiple subsystem operators are supposed to parallelly design individual retrofit controllers with accessibility only to respective subsystem models.
A retrofit controller is defined as a plug-in-type robust controller such that the stability of the resultant control system can be robustly assured for any possible variation of environments such that the original network system before implementing retrofit control is stable.
We aim at improving the resultant control performance while preserving the entire network stability.

In the line of our previous work, it is shown that such a retrofit controller from the standpoint of each subsystem operator can be designed without requiring any model of its environment; see \cite{ishizaki2018retrofit,sadamoto2018retrofit,inoue2018parameterization,sasahara2018parameterization} and references therein.
However, this, at the same time, implies that no information of an actual environment is used for  retrofit controller design.
Therefore, the resultant control performance is generally dependent on the possible variation of environments.
Such a low degree of freedom in the existing retrofit control could make a possible bottleneck for performance regulation, as will be demonstrated in this paper.
In fact, there remains a possibility to make use of some available information of environments to further improve the resultant control performance.

With this background, to reduce such a bottleneck in the existing retrofit control, we aim at developing a novel design method of retrofit controllers such that the resultant control performance can be regulated by adjusting the accuracy of approximate environment modeling.
In particular, we show that the stability of the resultant control system is robustly assured regardless of not only the stability of approximate environment models, but also the magnitude of modeling errors.
This robustness property of the proposed retrofit control is practically significant to incorporate existing  identification methods of unknown environments, such as in \cite{ljung1998system,bishop2006pattern}, because the accuracy of identified models may neither be reliable nor assurable in reality.
Furthermore, we conduct a control performance analysis to show that the foregoing bottleneck in the existing retrofit control can be reduced as improving the accuracy of environment modeling.
It should be noted that we do not explicitly discuss how to produce an approximate environment model in this paper, but we discuss how to effectively utilize an approximate environment model found by some offline identification before implementing retrofit control.

A distributed design method of decentralized controllers is developed in \cite{langbort2010distributed}, where the authors discuss the performance limitation of a linear quadratic regulator designed in a modular fashion.
This result is then generalized to the case of a network of multi-dimensional subsystems, the states of which are assumed to be fully controlled \cite{farokhi2013optimal}.
As an approach to modular design of decentralized controllers, a system decomposition method based on an integral quadratic constraint is developed in \cite{petes2017scalable}.
Though their formulation can actually frame a broad class of systems, conditions required for decomposed subsystems are generally conservative as remarked there.
As compared to these related approaches, the retrofit control has the advantage of applicability to more general stable network systems, for which we just assume the measurability of interconnection signals among subsystems.

We remark also that the retrofit control has a clear distinction from plug-and-play control \cite{stoustrup2009plug,bendtsen2013plug}, in which incremental addition of new devices, such as controllers, is considered for a working control system.
In general, the existing design schemes for plug-and-play control are not modular, meaning that an entire system model or its estimation is required for controller design.
From the viewpoint of modularity in design, we can also find a similarity with control system design based on passivity, or, more generally, dissipativity and passivity shortage \cite{willems1972dissipative,moylan1978stability,qu2014modularized}.
It is well known that negative feedback of passive subsystems retains the passivity.
This means that the stability of the entire network system can be ensured if individual subsystems are designed to be passive.
However, though a theoretically grounded procedure with modularity can be developed, the applicability of such a passivity-based approach is restrictive as compared to the retrofit control. 
This is simply because a network system of interest is not always decomposable into passive or passifiable subsystems.

The proposed retrofit control with approximate environment modeling is relevant to low-dimensional controller design based on model reduction \cite{antoulas2005approximation,obinata2001model,girard2009hierarchical}.
In particular, we can consider first applying model reduction to a system of interest, and then perform  controller design based on the resultant approximate model, where an approximation error due to model reduction can be handled as a model uncertainty in robust control \cite{zhou1995robust}.
However, such a model reduction method may not be applicable for a practical network system managed by multiple subsystem operators because a ``complete'' system model, to which model reduction is applied, is generally difficult to obtain.
In view of this, there are practical difficulties not only to find an approximate model, but also to assure the accuracy of approximate models.
The proposed retrofit control is a promising approach based on approximate modeling that can assure the control system stability without requiring the assurance of approximation accuracy.
We remark that such an unassured modeling error is not considered in a standard robust control setting.

The remainder of this paper is organized as follows.
In Section~2, we first review several existing results of retrofit control as a preliminary.
In particular, we give a motivating example that demonstrates a bottleneck for performance regulation in the existing retrofit control.
In Section~3, we develop a novel retrofit control method with approximate environment modeling. 
We first give a characterization of the new retrofit controllers by a frequency-domain analysis, based on the premise of a stability assumption of approximate environment models.
This stability assumption, making the Youla parameterization of retrofit controllers simpler, is made just for improving the visibility of a particular structure inside the retrofit controller.
Then, we provide a tractable state-space representation of approximate environment models the retrofit controllers by a time-domain analysis, which shows that the stability assumption of  premised in the frequency-domain analysis is not essential to prove the internal stability of the resultant control system.
Section~4 revisits the motivating example to show practical significance of the proposed retrofit control.
Finally, concluding remarks are provided in Section~5.

\textbf{Notation}~
The notation in this paper is generally standard: 
The identity matrix with an appropriate size is denoted by $I$.
The set of stable, proper, real rational transfer matrices is denoted by $\mathcal{RH}_{\infty}$.
For simplicity, all transfer matrices in the following are assumed to be proper and real rational.
The $\mathcal{L}_{\infty}$-norm of a transfer matrix $G$ with no singularities on the imaginary axis is denoted by $\|G\|_{\infty}$, which coincides with the $\mathcal{H}_{\infty}$-norm if $G$ is stable.
A transfer matrix $K$ is said to be a stabilizing controller for $G$ if the feedback system of $G$ and $K$ is internally stable in the standard sense \cite{zhou1995robust}.

\section{Review of Existing Retrofit Control}
\subsection{General Formulation}\label{secmf}
In this section, we first review several existing results of retrofit control reported in \cite{ishizaki2018retrofit,sadamoto2018retrofit,inoue2018parameterization,sasahara2018parameterization}.
Consider an interconnected linear system depicted in Fig.~\ref{figpresys}(a) where 
\begin{subequations}\label{prefed}
\begin{equation}\label{locsys}
\left[\hspace{-2pt}
\begin{array}{c}
w\\
z\\
y
\end{array}
\hspace{-2pt}\right]=
\underbrace{\left[\hspace{-2pt}
\begin{array}{ccc}
G_{wv}&G_{wd}&G_{wu}\\
G_{zv}&G_{zd}&G_{zu}\\
G_{yv}&G_{yd}&G_{yu}
\end{array}
\hspace{-2pt}\right]
}_{G}
\left[\hspace{-2pt}
\begin{array}{c}
v\\
d\\
u
\end{array}
\hspace{-2pt}\right]
\end{equation}
is referred to as a \textit{subsystem} of interest for retrofit control, and
\begin{equation}\label{envsys}
v=\overline{G}w
\end{equation}
is referred to as its \textit{environment}.
\end{subequations}
From the viewpoint of controlling a general network system composed of multiple subsystems,  $G$ corresponds to a particular subsystem for which a retrofit controller is designed by a subsystem operator, while the environment $\overline{G}$ corresponds to a lumped representation of all other subsystems, which can be high-dimensional.
In this formulation, the subsystem model of $G$ is assumed to be available for the retrofit controller design, while that of $\overline{G}$, which can be affected by other subsystem operators, is assumed to be unknown.

We denote the interconnection signals between the subsystem and its environment by $w$ and $v$, the evaluation output and disturbance input by $z$ and $d$, and the measurement output and control input by $y$ and $u$, respectively.
For the subsequent discussion, we use symbols denoting the submatrices of $G$, for example, as
\begin{equation}\label{defGs}
\begin{array}{rcl}
G_{(z,y)(d,u)}&:=&\left[\hspace{-2pt}
\begin{array}{cc}
G_{zd}&G_{zu}\\
G_{yd}&G_{yu}
\end{array}
\hspace{-2pt}\right],\quad
G_{(z,y)v}:=\left[\hspace{-2pt}
\begin{array}{cc}
G_{zv}\\
G_{yv}
\end{array}
\hspace{-2pt}\right],\vspace{2pt}\\
G_{w(d,u)}&:=&\left[\hspace{-2pt}
\begin{array}{cc}
G_{wd}&G_{wu}\\
\end{array}
\hspace{-2pt}\right].
\end{array}
\end{equation}
Then, we introduce the transfer matrix
$G_{\rm pre}:(d,u)\mapsto(z,y)$
defined by the feedback system of $G$ and $\overline{G}$ as
\begin{equation}\label{presys}
G_{\rm pre}:=G_{(z,y)(d,u)}+G_{(z,y)v}\overline{G}(I-G_{wv}\overline{G})^{-1}G_{w(d,u)}.
\end{equation}
We refer to $G_{\rm pre}$ as a \textit{preexisting system}, described as the dotted box in Fig.~\ref{figpresys}(a).
Based on this system description, the notion of retrofit controllers is defined as follows.

\begin{figure}[t]
\begin{center}
\includegraphics[width=85mm]{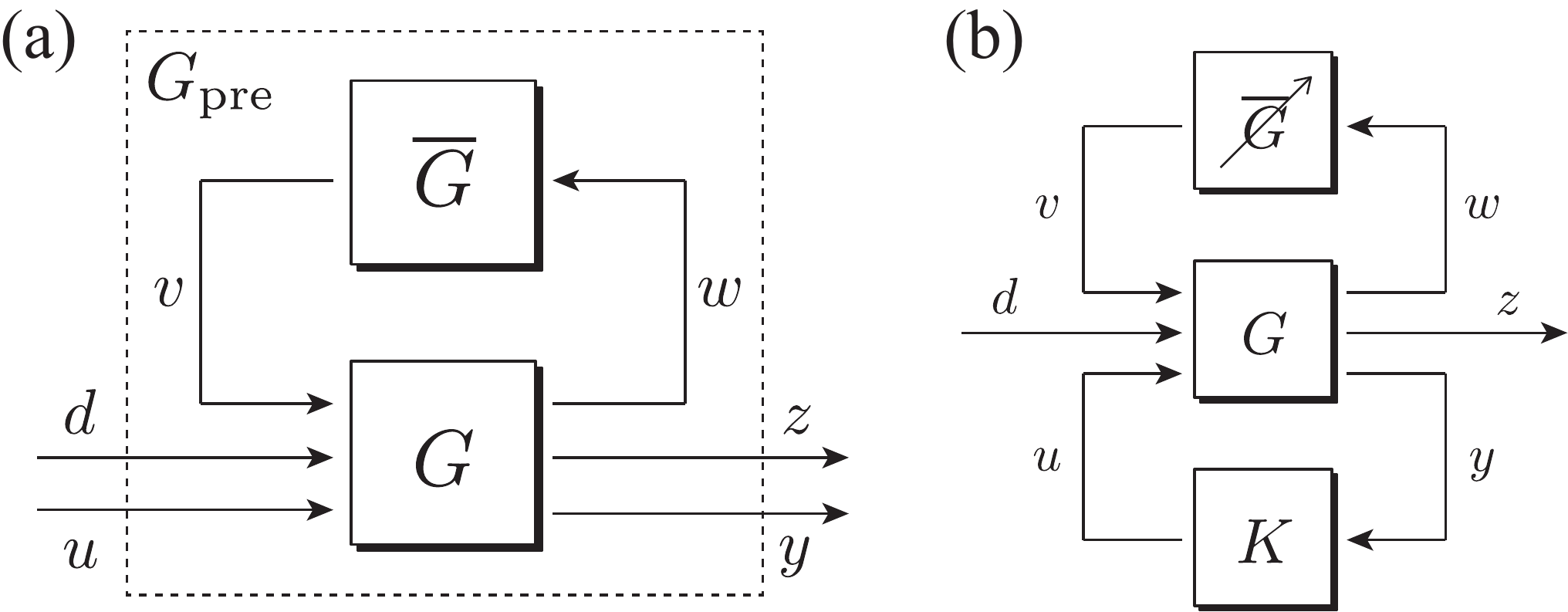}
\end{center}
\vspace{3pt}
\caption{
(a) Preexisting system composed of a subsystem of interest and its environment.
(b) Resultant control system.}
\label{figpresys}
\vspace{6pt}
\end{figure}

\begin{definition}\label{defretf}\normalfont
For the preexisting system $G_{\rm pre}$ in (\ref{presys}), define the set of all admissible environments as
\begin{equation}
\overline{\mathcal{G}}:=\left\{\overline{G}:G_{\rm pre}{\rm\ is\ internally\ stable}\right\}.
\end{equation}
An output feedback controller
\begin{equation}\label{cont}
u=Ky
\end{equation}
is said to be a \textit{retrofit controller} if the resultant control system in Fig.~\ref{figpresys}(b) is internally stable for any possible environment $\overline{G}\in\overline{\mathcal{G}}$.
\end{definition}

The retrofit controller is defined as a plug-in-type robust controller such that the stability of the resultant control system can be robustly assured for any possible variation of environments such that the preexisting system is stable.
We remark that the norm bound of the environment $\overline{G}$ is not premised.
Instead, we just premise the internal stability of the preexisting system.
Based on this definition, we first consider giving a parameterization of retrofit controllers.
To avoid unnecessary complication of controller parameterization based on the Youla parameterization \cite{youla1976modern,doyle2013feedback}, we make the following assumption.

\begin{assumption}\label{asstaG}\normalfont
The subsystem $G$ in (\ref{locsys}) is stable.
\end{assumption}

Assumption~\ref{asstaG} is made just for improving the visibility of a particular structure inside the retrofit controller, but it is not crucial to prove the resultant control system stability, as will be shown in Theorem~\ref{thmtdr}.
Then, we can derive the following parameterization of retrofit controllers.

\begin{proposition}\label{propgen}\normalfont
Let Assumption~\ref{asstaG} hold.
Consider the Youla parameterization of $K$ in (\ref{cont}) as
\begin{equation}\label{Ypara}
K=(I+QG_{yu})^{-1}Q,\quad Q\in \mathcal{RH}_{\infty}
\end{equation}
where $Q$ is the Youla parameter.
If 
\begin{equation}\label{cyoulao}
QG_{yv}=0,
\end{equation}
then $K$ is a retrofit controller.
\end{proposition}

Proposition~\ref{propgen} shows that the constrained version of the Youla parameterization in (\ref{cyoulao}) gives the parametrization of retrofit controllers.
We remark that, more generally, ``all" retrofit controllers can also be parameterized by the Youla parameterization such that 
\begin{equation}\label{gengqg}
G_{wu}QG_{yv}=0,
\end{equation}
which is more general than (\ref{cyoulao}).
This general parameterization, derived in \cite{inoue2018parameterization,sasahara2018parameterization}, further shows that the retrofit controller in Definition~\ref{defretf} can be characterized as a controller such that the transfer matrix from $v$ to $w$ of the local control system isolated from the environment is kept invariant.
This can be seen as follows.
Let $G_{wv}^{\prime}:v\mapsto w$ denote the transfer matrix from $v$ to $w$ in Fig.~\ref{figpresys}(b) as removing the block of $\overline{G}$.
Then, we have
\begin{equation}\label{intinv}
G_{wv}^{\prime}=G_{wv}+G_{wu}QG_{yv},
\end{equation}
which implies that $G_{wv}^{\prime}=G_{wv}$ for any retrofit controller because of (\ref{gengqg}).

In the rest of this paper, as following the terminology used in \cite{ishizaki2018retrofit,sadamoto2018retrofit,inoue2018parameterization,sasahara2018parameterization}, we refer to a retrofit controller $K$ parameterized in Proposition~\ref{propgen} as an \textit{output-rectifying retrofit controller}.
Note that such retrofit controllers are conditioned by the transfer matrices $G_{yu}$ and $G_{yv}$ that are relevant to the subsystem $G$, but not to the environment $\overline{G}$.
In other words, an output-rectifying retrofit controller can be designed only with the model information of $G$ isolated from $\overline{G}$.

\vspace{3pt}
\fboxrule=2\fboxrule
\begin{breakbox}
\noindent\underline{\textbf{NOTE}}~
One may think that the environment $\overline{G}$ can be regarded as a model uncertainty in robust control.
However, such a model uncertainty is typically assumed to be norm-bounded in a standard robust control setting.
In contrast, instead of assuming the norm bound of $\overline{G}$, the stability of the preexisting system, i.e., the stability of the feedback system of $G$ and $\overline{G}$ before controller implementation, is premised in the formulation of retrofit control.
In fact, this formulation leads to a particular class of controllers such that the interconnection transfer matrix is kept invariant as shown by (\ref{gengqg}) and (\ref{intinv}).
This further implies that the retrofit control does not aim at decoupling the subsystem $G$ from the environment $\overline{G}$, but it ``preserves" the dynamics with respect to the interconnection of $G$ and $\overline{G}$, the stability of which is premised as the preexisting system stability.
Therefore, we clearly see that the policy of retrofit control is essentially different from those of standard robust control \cite{zhou1995robust}, decoupling control \cite{falb1967decoupling}, and disturbance rejection (interconnection signal rejection) control \cite{wang2013distributed} in the literature.
\end{breakbox}\vspace{3pt}

\subsection{Review of Output-Rectifying Retrofit Control with Interconnection Signal Measurement}\label{secrevrc}

In this subsection, as a preliminary for the main theoretical developments in Section~\ref{sectd}, we describe specific results on output-rectifying retrofit control in Proposition~\ref{propgen}.
These results can be derived as a simple generalization of results in \cite{ishizaki2018retrofit,sadamoto2018retrofit,inoue2018parameterization,sasahara2018parameterization}, but are not exactly the same as those.
In the rest of this paper, we consider the following situation.

\begin{assumption}\label{assumvy}\normalfont
The interconnection signal $v$ and $w$ are measurable in addition to the measurement output $y$ in (\ref{prefed}).
\end{assumption}

From a symbolic viewpoint, Assumption~\ref{assumvy} corresponds to the situation where every symbol $y$ in the discussion of Section~\ref{secmf} is to be replaced with the new measurement output $(y,w,v)$.
Based on this premise, the transfer matrices in (\ref{locsys}) relevant to $y$ are also redefined.
For example, $G_{yv}$ and $G_{yu}$ are redefined as
\[
G_{(y,w,v)v}:=\left[\hspace{-2pt}
\begin{array}{c}
G_{yv}\\
G_{wv}\\
I
\end{array}
\hspace{-2pt}\right],\quad
G_{(y,w,v)u}:=\left[\hspace{-2pt}
\begin{array}{c}
G_{yu}\\
G_{wu}\\
0
\end{array}
\hspace{-2pt}\right].
\]
Furthermore, the controller $K$ in (\ref{cont}) is also redefined as
\begin{equation}\label{retcont}
u=\underbrace{\left[\hspace{-2pt}
\begin{array}{ccc}
K_{y}&K_{w}&K_{v}
\end{array}
\hspace{-2pt}\right]}_{K}
\left[\hspace{-2pt}
\begin{array}{c}
y\\
w\\
v
\end{array}
\hspace{-2pt}\right],
\end{equation}
the Youla parameterization of which can be written as
\[
K=(I+QG_{(y,w,v)u})^{-1}Q,\quad Q\in \mathcal{RH}_{\infty}.
\]
A simple but notable fact to be used is that 
\begin{equation}\label{outrec}
R:=\left[\hspace{-2pt}
\begin{array}{ccc}
I&0&-G_{yv}\\
0&I&-G_{wv}\end{array}
\hspace{-2pt}\right]
\end{equation}
is a basis of the left kernel of $G_{(y,w,v)v}$ in $\mathcal{RH}_{\infty}$, i.e., 
\begin{equation}\label{factQ}
\begin{array}{l}
QG_{(y,w,v)v}=0,\quad Q\in \mathcal{RH}_{\infty}\\
\hspace{60pt}\Longleftrightarrow \quad\exists\hat{Q}\in \mathcal{RH}_{\infty}\ \ {\rm s.t.}\ \ Q=\hat{Q}R.
\end{array}
\end{equation}
Using this left kernel basis $R$, we can rewrite (\ref{Ypara}) and (\ref{cyoulao}) as
\[
K=\hat{K}R,\quad\hat{K}=(I+\hat{Q}G_{(y,w)u})^{-1}\hat{Q},\quad\hat{Q}\in \mathcal{RH}_{\infty},
\]
where we have used the fact that
\begin{equation}\label{GRG}
G_{(y,w)u}=RG_{(y,w,v)u}.
\end{equation}
We refer to $R:(y,w,v)\mapsto(\hat{y},\hat{w})$ as an \textit{output rectifier}, the name of which is based on the fact that the measurement output $(y,w,v)$ is rectified in such a way that
\begin{equation}\label{orect}
\hat{y}=y-G_{yv}v,\quad\hat{w}=w-G_{wv}v.
\end{equation}
The rectified output $(\hat{y},\hat{w})$ is used as the input of $\hat{K}$, which is referred to as a \textit{module controller}.
This discussion leads to an ``explicit'' representation of all output-rectifying retrofit controllers as follows.

\begin{proposition}\label{propost}\normalfont
Let Assumptions~\ref{asstaG} and \ref{assumvy} hold.
Then, $K$ in (\ref{retcont}) is an output-rectifying retrofit controller if and only if
\begin{equation}\label{stoutret}
K=\hat{K}R
\end{equation}
where $\hat{K}$ is a stabilizing controller for $G_{(y,w)u}$ and $R$ is defined as in (\ref{outrec}).
Furthermore, the block diagram of the resultant control system is depicted as in Fig.~\ref{figorigsigflow}, i.e., the entire map $T_{zd}:d\mapsto z$ is given as
\begin{equation}
T_{zd}=\hat{T}_{zd}(\hat{K})+G_{zv}(I-\overline{G}G_{wv})^{-1}\overline{G}\hat{T}_{wd}(\hat{K})
\end{equation}
where $\hat{T}_{zd}:d\mapsto\hat{z}$ and $\hat{T}_{wd}:d\mapsto\hat{w}$ denote the transfer matrices compatible with Fig.~\ref{figorigsigflow}, given as
\begin{equation}
\begin{array}{rcl}
\hat{T}_{zd}(\hat{K})&:=&G_{zd}+G_{zu}\hat{K}(I-G_{(y,w)u}\hat{K})^{-1}G_{(y,w)d}\\
\hat{T}_{wd}(\hat{K})&:=&G_{wd}+G_{wu}\hat{K}(I-G_{(y,w)u}\hat{K})^{-1}G_{(y,w)d}.
\end{array}
\end{equation}
\end{proposition}

\begin{figure}[t]
\begin{center}
\includegraphics[width=80mm]{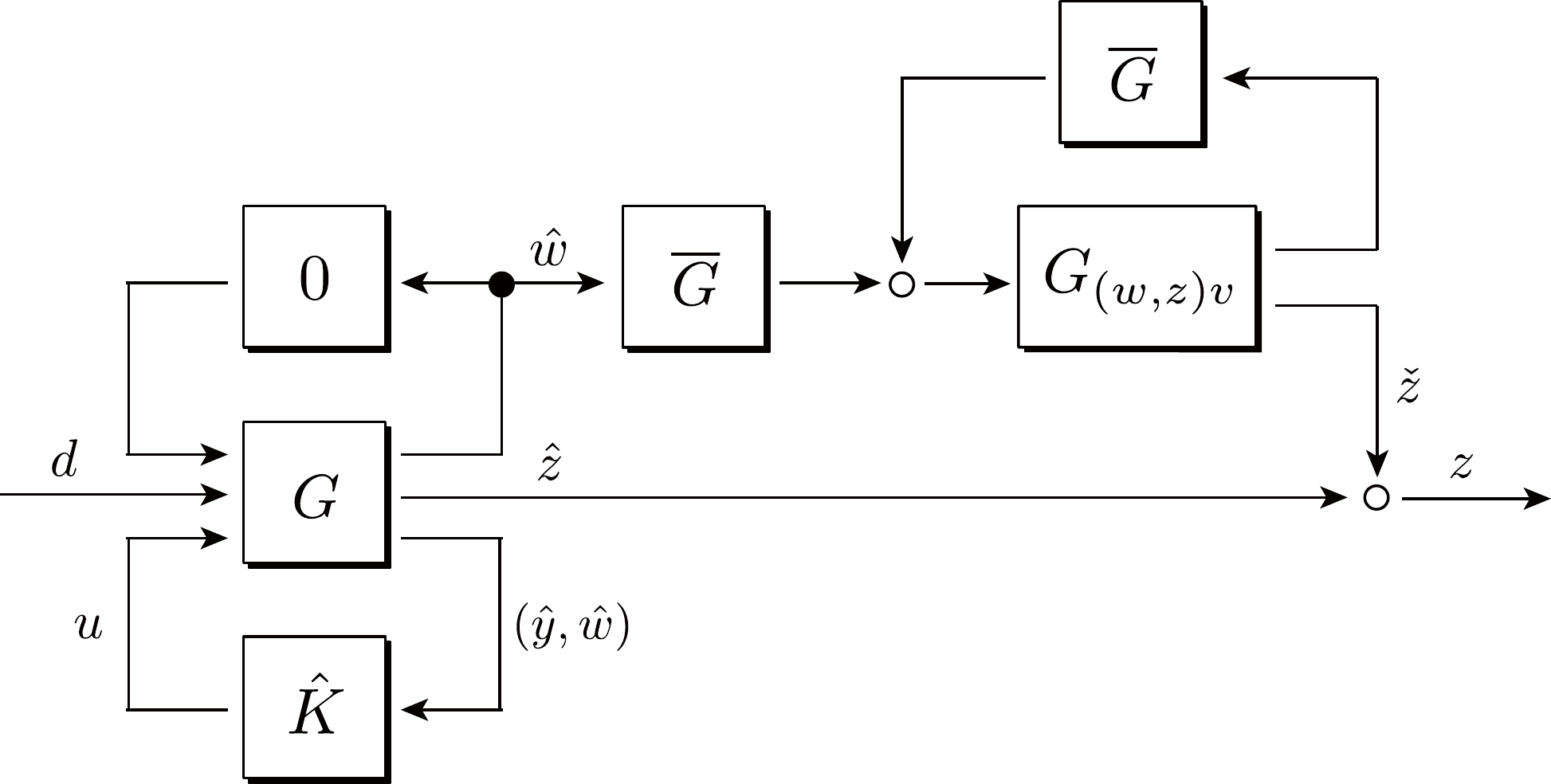}
\end{center}
\vspace{-0pt}
\caption{
Block diagram of existing retrofit control.}
\label{figorigsigflow}
\vspace{3pt}
\end{figure}

Proposition~\ref{propost} can be proven as a special case of Theorem~\ref{thmost} below.
It is shown that, if the module controller $\hat{K}$ is designed as a stabilizing controller for $G_{(y,w)u}$, which is isolated from $\overline{G}$, then the stability of the resultant control system is always assured by $K$ in (\ref{stoutret}).
This clearly shows the modularity of retrofit controller design; the model information of $\overline{G}$ is not required to assure, at least, the stability of the resultant control system.

Next, we analyze the resultant control performance.
From the block diagram in Fig.~\ref{figorigsigflow}, we see that $z$ can be decomposed as
$z=\hat{z}+\check{z}$ where
\[
\hat{z}=\hat{T}_{zd}(\hat{K})d,\quad
\check{z}=G_{zv}(I-\overline{G}G_{wv})^{-1}\overline{G}\hat{T}_{wd}(\hat{K})d.
\]
The triangular inequality for the induced norm of $z$ leads to the following upper and lower bounds of the resultant control performance.

\begin{proposition}\label{propbnd}\normalfont
With the same notation as that in Proposition~\ref{propost}, the resultant control performance is bounded as
\begin{equation}\label{Tul}
|\check{\gamma}-\hat{\gamma}|\leq\|T_{zd}\|_{\infty}\leq\hat{\gamma}+\check{\gamma}
\end{equation}
where the induced gains of $\hat{z}$ and $\check{z}$ are given as
\[
\hat{\gamma}:=\|\hat{T}_{zd}(\hat{K})\|_{\infty},\quad \check{\gamma}:=\|G_{zv}(I-\overline{G}G_{wv})^{-1}\overline{G}\hat{T}_{wd}(\hat{K})\|_{\infty}.
\]
\end{proposition}

We remark that $\hat{\gamma}$ is ``directly regulatable'' by a suitable choice of $\hat{K}$, but $\check{\gamma}$ is not because the term dependent on $\overline{G}$ is involved.
For explanation, let us consider a situation where $\hat{\gamma}$ is made sufficiently small, but $\check{\gamma}$ is not, i.e., $\hat{\gamma}\ll\check{\gamma}$.
Then, (\ref{Tul}) implies that $\|T_{zd}\|_{\infty}\simeq\check{\gamma}$.
This means that actual control performance may not be satisfactory, even if $\hat{\gamma}$ is regulated desirably.
Such an undesirable situation possibly arises when the magnitude of $\overline{G}$ is large.

From the observation above, we can see that a large value of $\check{\gamma}$ makes a ``bottleneck'' to perform satisfactory regulation based on the existing retrofit control.
We can say that the value of $\check{\gamma}$ evaluates a gap between $\|T_{zd}\|_{\infty}$ and $\hat{\gamma}$, each of which corresponds to the ``actual performance level'' of the resultant control system and the ``assumed performance level'' of the modular control system.
The simplest but not realistic situation for the minimum gap is $\overline{G}=0$, which leads to the ideal situation where $z=\hat{z}$, or equivalently, $\|T_{zd}\|_{\infty}=\hat{\gamma}$, i.e., the actual performance level is equal to the assumed performance level.

\subsection{Motivating Example}\label{secmot}

\begin{figure}[t]
\begin{center}
\includegraphics[width=80mm]{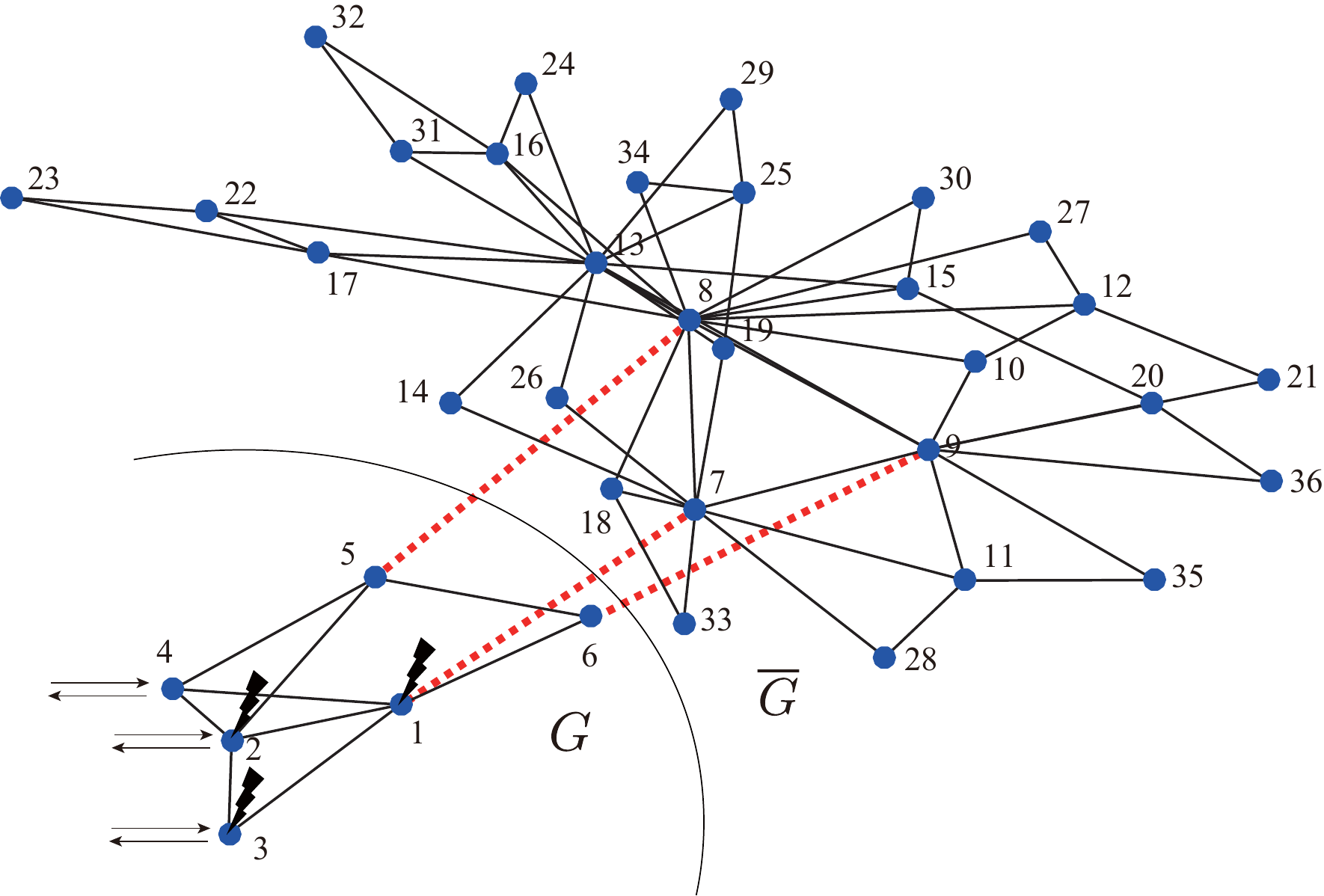}
\end{center}
\vspace{-0pt}
\caption{
Second-order oscillator network composed of 36 nodes. 
The subnetwork of the first six nodes corresponds to a subsystem of interest, and the remaining part is its unknown environment.
}
\label{fignetwork}
\vspace{3pt}
\end{figure}

We give a motivating example that demonstrates the bottleneck of the existing retrofit control described in Section~\ref{secrevrc}, towards highlighting the main contribution of this paper.
Consider a network system composed of 36 nodes depicted in Fig~\ref{fignetwork}.
The subnetwork of the first six nodes is supposed to be a subsystem $G$ of interest, and the remaining part is supposed to be its unknown environment $\overline{G}$.
Let $\mathcal{I}$ and $\overline{\mathcal{I}}$ denote the label sets corresponding to $G$ and $\overline{G}$, i.e.,
\[
\mathcal{I}=\{1,2,\ldots,6\},\quad\overline{\mathcal{I}}=\{7,8,\ldots,36\}.
\]
Furthermore, let $\mathcal{N}_{i}$ denote the label set corresponding to the set of nodes such that they are adjacent to the $i$th node and involved in $\mathcal{I}$. 
In a similar fashion, let $\overline{\mathcal{N}_{i}}$ denote the label set of the other adjacent nodes involved in $\overline{\mathcal{I}}$.
With this notation, for each $i\in \mathcal{I}$, the node dynamics is given as
\begin{equation}\label{secsys}
M_{i}\ddot{\theta}_{i}+D_{i}\dot{\theta}_{i}+\sum_{j\in \mathcal{N}_{i}}K_{ij}(\theta_{j}-\theta_{i})+v_{i}=u_{i}+d_{i}
\end{equation}
where $\theta_{i}$ denotes the angular state, $u_{i}$ denotes the control input, $d_{i}$ denotes the disturbance input, and
\begin{equation}\label{intvi}
v_{i}=\sum_{j\in\overline{\mathcal{N}}_{i}}K_{ij}(\theta_{j}-\theta_{i})
\end{equation}
denotes the interconnection signal from the environment.
The node dynamics of the environment is given in the same fashion without the terms of $u_{i}$ and $d_{i}$.
The second-order oscillator network (\ref{secsys}) can be regarded as a mechanical analog of synchronous generators \cite{ishizaki2018graph}.
In the context of power system modeling, the interconnection signal $v_{i}$ in (\ref{intvi}) corresponds to the power flow between the subsystem and its environment.
The three interconnection links are depicted by the dotted lines in Fig~\ref{fignetwork}.

In the following simulation, we set all the inertia constants and damping constants as $M_{i}=1$ and $D_{i}=0.2$.
Furthermore, we set the coupling constants inside the subsystem and inside the environment uniformly as
\[
\begin{array}{l}
K_{ij}=5,\quad\forall j\in \mathcal{N}_{i};\quad i\in \mathcal{I},\\
K_{ij}=5,\quad\forall j\in\overline{\mathcal{N}_{i}};\quad i\in\overline{\mathcal{I}}.
\end{array}
\]
The coupling constants between the subsystem and environment are to be varied as a parameter $k_{\rm c}$, i.e.,
\begin{equation}\label{kc}
K_{ij}=k_{\rm c},\quad\forall j\in\overline{\mathcal{N}_{i}};\quad i\in \mathcal{I}.
\end{equation}
For simplicity, we assume the symmetry $K_{ij}=K_{ji}$.

For retrofit controller design, the control input and the disturbance input are assigned as
\[
u=(u_{i})_{i\in\{2,3,4\}},\quad d=(d_{i})_{i\in\{1,2,3\}},
\]
respectively.
The measurement output and the evaluation output are assigned as
\[
y=(\theta_{i},\dot{\theta}_{i})_{i\in\{2,3,4\}},\quad z=(\dot{\theta}_{i})_{i\in\{1,2,\ldots,6\}},
\]
respectively.
In addition to $y$, the interconnection signals 
\[
v=(v_{i})_{i\in\{1,5,6\}},\quad w=(\theta_{i})_{i\in\{1,5,6\}}
\]
are assumed to be measurable.
We remark that only the local model parameters $K_{ij}$ for $(i,j)\in \mathcal{I}\times \mathcal{I}$ and $M_{i},\ D_{i}$ for $i\in \mathcal{I}$ are assumed to be available.
The entire network system is originally stable for any nonnegative value of $k_{\rm c}$ in (\ref{kc}). 
Though the system has  a single zero eigenvalue, it does not matter because the corresponding eigenspace is unobservable from the evaluation output $z$.

\begin{figure*}[t]
\begin{center}
\includegraphics[width=140mm]{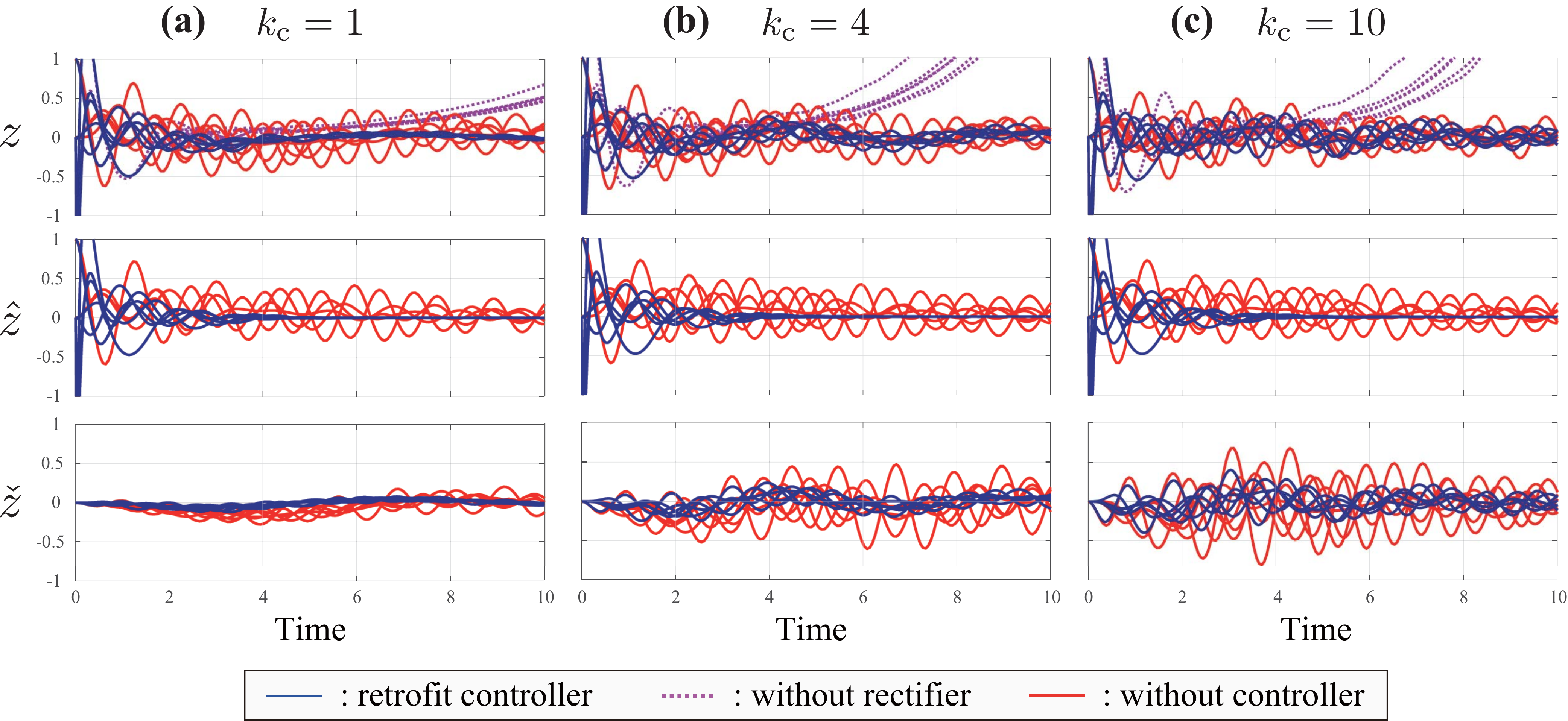}
\end{center}
\vspace{3pt}
\caption{
Resultant control system behavior in response to impulsive disturbance at the first node.
(a) Case of weak coupling between $G$ and $\overline{G}$. 
(b) Case of moderate coupling between $G$ and $\overline{G}$.
(c) Case of strong coupling between $G$ and $\overline{G}$.}
\label{figtrajml}
\vspace{6pt}
\end{figure*}

For the design of the module controller $\hat{K}$ in (\ref{stoutret}), we apply the standard $\mathcal{H}_{\infty}$-control synthesis technique to $G$, isolated from $\overline{G}$, such that
\begin{equation}\label{Jalp}
J_{\alpha}=\displaystyle \sup_{d\in \mathcal{L}_{2}\backslash\{0\}}\frac{\|(z,\alpha u)\|_{\mathcal{L}_{2}}}{\|d\|_{\mathcal{L}_{2}}}
\end{equation}
is minimized where $\alpha$ is a weighting constant for the control input.
Setting $\alpha=0.2$, we plot the impulse response of the resultant control system in Fig.~\ref{figtrajml}, where Figs.~\ref{figtrajml}(a)-(c) correspond to the cases of $k_{\rm c}=1$,  $k_{\rm c}=4$, and $k_{\rm c}=10$, i.e., weak coupling, moderate coupling, and strong coupling, respectively.
In each top subfigure of Figs.~\ref{figtrajml}(a)-(c), the blue solid lines show the trajectory of $z$ when the retrofit controller $K$ in (\ref{stoutret}) is used, the red solid lines show the case where no controller is used, and the magenta dotted lines show the case where the output rectifier $R$ is not involved in the controller, i.e., the module controller $\hat{K}$ is directly implemented as a simple decentralized controller
\begin{equation}\label{worec}
u=\hat{K}\left[\hspace{-2pt}
\begin{array}{c}
y\\
w
\end{array}
\hspace{-2pt}\right].
\end{equation}
From these top subfigures, we see that the direct implementation of $\hat{K}$ induces the  instability of the resultant control systems even though $\hat{K}$ is designed to be a stabilizing controller for $G$.
In contrast, the retrofit controller can actually guarantee the stability of the resultant control system for all the values of the coupling constant $k_{\rm c}$.

However, we can also see that the amplitude of $z$ becomes larger as the coupling between $G$ and $\overline{G}$ becomes stronger.
This outcome can be explained as follows.
The decomposed outputs $\hat{z}$ and $\check{z}$ in Fig.~\ref{figorigsigflow} are plotted in the middle and the bottom of Figs.~\ref{figtrajml}(a)-(c), where the blue and red lines correspond to the cases with and without the retrofit controller, respectively.
Note that the actual output $z$ is equal to the sum of $\hat{z}$ and $\check{z}$, the induced gains of which are denoted by $\hat{\gamma}$ and $\check{\gamma}$ in (\ref{Tul}), respectively.
In fact, the behavior of $\hat{z}$ is well controlled, and it is identical for all the values of $k_{\rm c}$ because the module controller $\hat{K}$ is designed only with the information of $G$, which is not dependent on $k_{\rm c}$.
In contrast, the magnitude of $\check{z}$ is amplified as $k_{\rm c}$ increases, i.e., as the gain of $\overline{G}$ increases.
In accordance with this amplification, the magnitude of the resultant $z$ is also amplified.

\begin{figure}[t]
\begin{center}
\includegraphics[width=80mm]{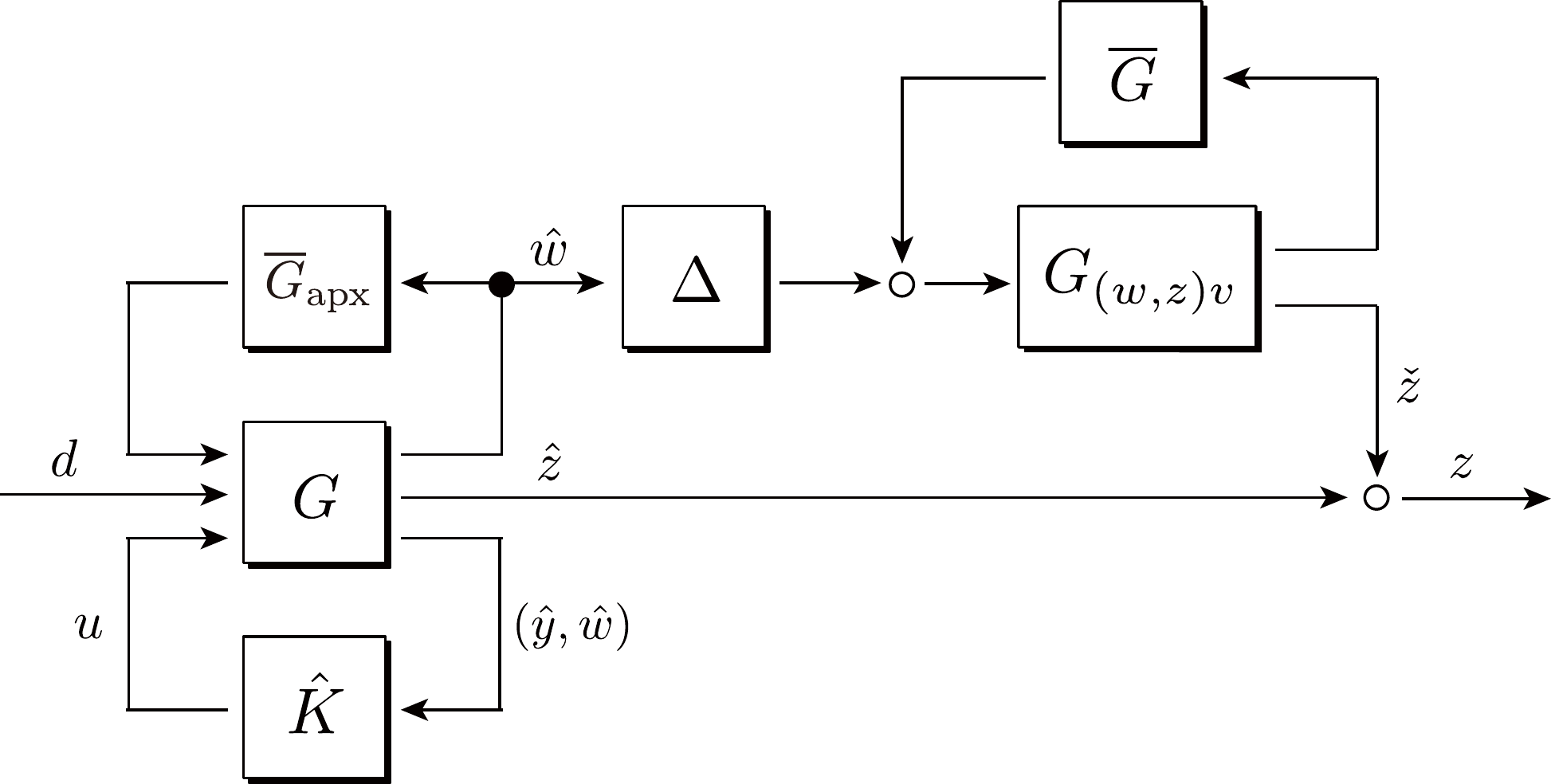}
\end{center}
\vspace{-0pt}
\caption{
Block diagram of extended retrofit control.}
\label{fignewsigflow}
\vspace{3pt}
\end{figure}

As demonstrated here, a small-gain property for $\overline{G}$ may be required for satisfactory regulation, though the internal stability of the resultant control system can be assured for any possible $\overline{G}$.
This is mainly because only an ``identical" retrofit controller designed with $G$ is used regardless of the variation of $\overline{G}$.
To overcome this drawback, in the following sections, we will develop a new retrofit control method that can produce the block diagram in Fig.~\ref{fignewsigflow}, where we make use of an approximate model $\overline{G}_{\rm apx}$ of the environment $\overline{G}$. 
An important difference between Fig.~\ref{figorigsigflow} and Fig.~\ref{fignewsigflow} is that the block of $\overline{G}$ in the middle of Fig.~\ref{figorigsigflow} is replaced with the block of
\begin{equation}\label{modeler}
\Delta:=\overline{G}-\overline{G}_{\rm apx}
\end{equation}
in Fig.~\ref{fignewsigflow}. 
Note that $\Delta$ represents a modeling error because $\overline{G}_{\rm apx}$ represents an approximate model of $\overline{G}$. 
Intuitively, as making the modeling error $\Delta$ small, we can generally reduce the amplitude of $\check{z}$.
We remark that the norm bound of $\Delta$ is assumed not to be assurable because $\overline{G}$ is assumed to be unknown.
In the next section, we will develop such an extended version of the retrofit control method that assures the entire system stability without assuming any assurance of environment modeling accuracy.

\section{Theoretical Developments}\label{sectd}
\subsection{Frequency-Domain Analysis: Characterization of Extended Retrofit Controllers}\label{secfda}

In this section, we premise that an approximate environment model has been found in some way, though its modeling accuracy is not assured for retrofit controller design.
Our basic strategy to incorporate such unassured environment modeling is to regard the feedback of the subsystem $G$ and the approximate environment model $\overline{G}_{\rm apx}$ as a new subsystem of interest.
This corresponds to the situation where the original preexisting system Fig.~\ref{figpresys}(a) is equivalently regarded as the feedback system in Fig.~\ref{fignewpre}.
In particular, we regard 
\begin{equation}\label{defGpls}
G^{+}:=G+G_{(w,z,y)v}\overline{G}_{\rm apx}(I-G_{wv}\overline{G}_{\rm apx})^{-1}G_{w(v,d,u)}
\end{equation}
as a new subsystem of interest and the modeling error $\Delta$ in (\ref{modeler}) as a new environment.
In this formulation, it is interesting to note that the modeling error $\Delta$ can be viewed as a dynamical component that stabilizes the new subsystem $G^{+}$.
Clearly, $G^{+}=G$ holds if $\overline{G}_{\rm apx}=0$.

In the following discussion, in a manner similar to (\ref{defGs}), we denote submatrices of $G^{+}$, e.g., by
\[
G_{(y,w)u}^{+}:=\left[\hspace{-2pt}
\begin{array}{cc}
G_{yu}^{+}\\
G_{wu}^{+}
\end{array}
\hspace{-2pt}\right],\quad
G_{(y,w)d}^{+}:=\left[\hspace{-2pt}
\begin{array}{cc}
G_{yd}^{+}\\
G_{wd}^{+}
\end{array}
\hspace{-2pt}\right].
\]
One may think that the existing results in Section~\ref{secrevrc} can be directly applied as simply replacing $G$ with $G^{+}$, and $\overline{G}$ with $\Delta$.
However, it is not very clear to see if such a simple replacement is valid or not because the interconnection signals between $G^{+}$ and $\Delta$ are found to be $v-\overline{G}_{\rm apx}w$ and $w$, which are clearly different from the original interconnection signals $v$ and $w$ between $G$ and $\overline{G}$.
Therefore, we need to carefully discuss how $K$ in the form of (\ref{retcont}) should be modified or generalized in this new formulation of retrofit control.
As an answer to this question, we will show that the set of all retrofit controllers with environment modeling actually coincides with the set of all retrofit controllers in Proposition~\ref{propost}, but has a much complicated structure.

\begin{figure}[t]
\begin{center}
\includegraphics[width=40mm]{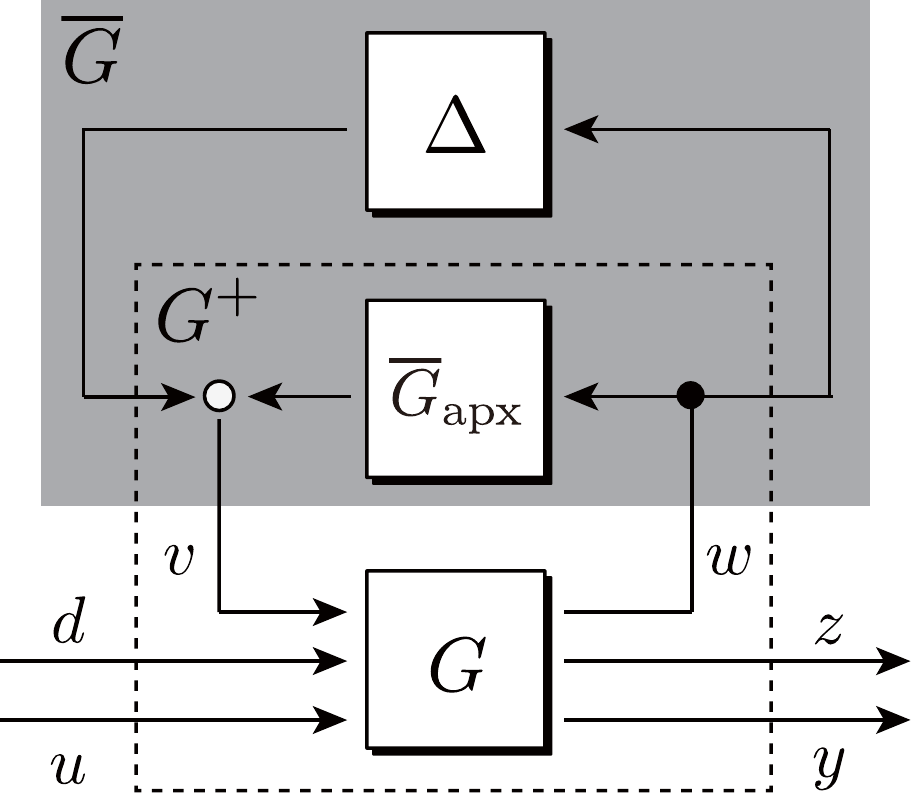}
\end{center}
\vspace{3pt}
\caption{
An equivalent representation of preexisting system.}
\label{fignewpre}
\vspace{6pt}
\end{figure}

In the derivation of Proposition~\ref{propost}, we started the discussion from the fact that the Youla parameter $Q$ can be factorized as in (\ref{factQ}), and then we showed that $\hat{K}$ in (\ref{stoutret}) is found to be a stabilizing controller for $G_{(y,w)u}$.
In what follows, as a converse direction, we first suppose that $\hat{K}$ is given as a stabilizing controller for $G_{(y,w)u}^{+}$, and then we will derive a compatible factorization of $Q$.
To make the Youla parameterization tractable, we make the following assumption.

\begin{assumption}\label{asmodsta}\normalfont
The approximate model $\overline{G}_{\rm apx}$ belongs to $\overline{\mathcal{G}}$, i.e., $G^{+}$ in (\ref{defGpls}) is internally stable.
\end{assumption}

Assumption~\ref{asmodsta} is in fact not crucial to prove the resultant control system stability, as shown in Theorem~\ref{thmtdr} below.
Owing to this assumption, the Youla parameterization of $\hat{K}$ can be simply written as
\begin{equation}\label{parakhat}
\hat{K}=(I+\hat{Q}G_{(y,w)u}^{+})^{-1}\hat{Q},\quad\hat{Q}\in \mathcal{RH}_{\infty}.
\end{equation}
This means that $\hat{K}$ is a stabilizing controller for $G_{(y,w)u}^{+}$.
As a generalization of (\ref{GRG}), we notice that
\begin{equation}\label{Gpfac}
G_{(y,w)u}^{+}=XRG_{(y,w,v)u}
\end{equation} 
where $X$, being invertible in $\mathcal{RH}_{\infty}$, is defined as
\begin{equation}\label{defX}
X:=\left[\hspace{-2pt}
\begin{array}{cc}
I&G_{yv}\overline{G}_{\rm apx}(I-G_{wv}\overline{G}_{\rm apx})^{-1}\\
0&(I-G_{wv}\overline{G}_{\rm apx})^{-1}
\end{array}
\hspace{-2pt}\right].
\end{equation}
Note that $X\in \mathcal{RH}_{\infty}$ for any $\overline{G}_{\rm apx}\in\overline{\mathcal{G}}$.
In addition, $X=I$ if $\overline{G}_{\rm apx}=0$.
Substituting (\ref{Gpfac}) into (\ref{parakhat}) and multiplying it by $XR$ from the right side, we have
\begin{equation}\label{Kparan}
\underbrace{\hat{K}XR}_{K}=(I+\underbrace{\hat{Q}XR}_{Q}G_{(y,w,v)u})^{-1}\underbrace{\hat{Q}XR}_{Q},\quad Q\in \mathcal{RH}_{\infty},
\end{equation}
which gives the Youla parameterization of $K$ such that (\ref{factQ}) holds.
We remark that $XR:(y,w,v)\mapsto(\hat{y},\hat{w})$ can be seen as an extended output rectifier that performs the  output rectification of
\begin{equation}\label{recout}
\begin{array}{ccl}
\hat{y}&=&(y-G_{yv}v)+G_{yv}\overline{G}_{\rm apx}(I-G_{wv}\overline{G}_{\rm apx})^{-1}(w-G_{wv}v),\\
\hat{w}&=&(I-G_{wv}\overline{G}_{\rm apx})^{-1}(w-G_{wv}v),
\end{array}
\end{equation}
which is a generalization of (\ref{orect}).
This derivation enables to generalize Proposition~\ref{propost} as follows.

\begin{theorem}\label{thmost}\normalfont
Let Assumptions~\ref{asstaG}, \ref{assumvy}, and \ref{asmodsta} hold.
Then, $K$ in (\ref{retcont}) is an output-rectifying retrofit controller if and only if
\begin{equation}\label{stoutretn}
K=\hat{K}XR
\end{equation}
where $\hat{K}$ is a stabilizing controller for $G_{(y,w)u}^{+}$, and $R$ and $X$ are defined as in (\ref{outrec}) and (\ref{defX}), respectively.
Furthermore, the block diagram of the resultant control system is depicted as in Fig.~\ref{fignewsigflow}, i.e., the entire map $T_{zd}:d\mapsto z$ is given as
\begin{equation}\label{casTzd}
T_{zd}=\hat{T}_{zd}^{+}(\hat{K})+G_{zv}(I-\overline{G}G_{wv})^{-1}\Delta\hat{T}_{wd}^{+}(\hat{K})
\end{equation}
where $\hat{T}_{zd}^{+}:d\mapsto\hat{z}$ and $\hat{T}_{wd}^{+}:d\mapsto\hat{w}$ denote the transfer matrices compatible with Fig.~\ref{fignewsigflow}, given as
\begin{equation}
\begin{array}{rcl}
\hat{T}_{zd}^{+}(\hat{K})&:=&G_{zd}^{+}+G_{zu}^{+}\hat{K}(I-G_{(y,w)u}^{+}\hat{K})^{-1}G_{(y,w)d}^{+}\\
\hat{T}_{wd}^{+}(\hat{K})&:=&G_{wd}^{+}+G_{wu}^{+}\hat{K}(I-G_{(y,w)u}^{+}\hat{K})^{-1}G_{(y,w)d}^{+},
\end{array}
\end{equation}
and $\Delta$ is defined as in (\ref{modeler}).
\end{theorem}

\begin{proof}\normalfont
We see that (\ref{Kparan}) is the Youla parameterization of $K$ in (\ref{stoutretn}).
Note that (\ref{parakhat}) is equivalent to (\ref{Kparan}) because $X$ is invertible and $R$ is right invertible.
Thus, all output-rectifying retrofit controllers in the form of (\ref{retcont}) can be written as (\ref{stoutretn}).

Next, let us prove (\ref{casTzd}).
As shown in \cite{sasahara2018parameterization}, for any output-rectifying retrofit controller, the entire map is given as
\[
\begin{array}{l}
T_{zd}=G_{zd}+G_{zu}QG_{(y,w,v)d}\\
\hspace{20pt}+G_{zv}(I-\overline{G}G_{wv})^{-1}\overline{G}(G_{wd}+G_{wu}QG_{(y,w,v)d}).
\end{array}
\]
In a similar manner to (\ref{Gpfac}), we have
\[
G_{(y,w)d}^{+}=XRG_{(y,w,v)d}.
\]
Thus, we see that
\[
\begin{array}{l}
T_{zd}=G_{zd}+G_{zu}\hat{Q}G_{(y,w)d}^{+}\\
\hspace{20pt}+G_{zv}(I-\overline{G}G_{wv})^{-1}\overline{G}(G_{wd}+G_{wu}\hat{Q}G_{(y,w)d}^{+}).
\end{array}
\]
On the other hand, the input-to-output map of Fig.~\ref{fignewsigflow}, denoted here by 
$T_{zd}^{\prime}:d\mapsto z$, is given as
\[
\begin{array}{l}
T_{zd}^{\prime}=G_{zd}^{+}+G_{zu}^{+}\hat{Q}G_{(y,w)d}^{+}\\
\hspace{20pt}+G_{zv}(I-\overline{G}G_{wv})^{-1}\Delta(G_{wd}^{+}+G_{wu}^{+}\hat{Q}G_{(y,w)d}^{+}).
\end{array}
\]
For the identity of $T_{zd}=T_{zd}^{\prime}$, it suffices to show that
\begin{equation}\label{GGpls}
\begin{array}{l}
G_{zd}+G_{zv}(I-\overline{G}G_{wv})^{-1}\overline{G}G_{wd}\\
\hspace{40pt}
=G_{zd}^{+}+G_{zv}(I-\overline{G}G_{wv})^{-1}\Delta G_{wd}^{+}
\end{array}
\end{equation}
and
\[
\begin{array}{l}
G_{zu}+G_{zv}(I-\overline{G}G_{wv})^{-1}\overline{G}G_{wu}\\
\hspace{40pt}
=G_{zu}^{+}+G_{zv}(I-\overline{G}G_{wv})^{-1}\Delta G_{wu}^{+}.
\end{array}
\]
Because both equalities can be proven in a similar manner, we only prove (\ref{GGpls}).
Subtracting the left-hand side of (\ref{GGpls}) from the right-hand side, we have
\[
\begin{array}{l}
G_{zv}\biggl[\bigl\{\underbrace{I-(I-\overline{G}G_{wv})^{-1}}_{-(I-\overline{G}G_{wv})^{-1}\overline{G}G_{wv}}\bigr\}\overline{G}_{\rm apx}(I-G_{wv}\overline{G}_{\rm apx})^{-1}\\
\hspace{20pt}
+(I-\overline{G}G_{wv})^{-1}\overline{G}
\bigl\{\underbrace{(I-G_{wv}\overline{G}_{\rm apx})^{-1}-I}_{G_{wv}\overline{G}_{\rm apx}(I-G_{wv}\overline{G}_{\rm apx})^{-1}}\bigr\}\biggr]G_{wd}=0.
\end{array}
\]
The relations indicated by the underbraces are proven by
\begin{equation}\label{relPK1}
\begin{array}{ccl}
(I+PK)^{-1}&=&I+PK(I-PK)^{-1}\\
&=&I+(I-PK)^{-1}PK.
\end{array}
\end{equation}
Hence, the claim is proven. \hfill $\square $
\end{proof}

Theorem~\ref{thmost} provides another representation of all output-rectifying retrofit controllers in which the approximate environment model is involved as a tuning parameter.
In particular, $K$ in (\ref{stoutretn}) is shown to be an output-rectifying retrofit controller if the module controller
\begin{subequations}\label{modcons}
\begin{equation}
u=\underbrace{\left[\hspace{-2pt}
\begin{array}{ccc}
\hat{K}_{y}&\hat{K}_{w}
\end{array}
\hspace{-2pt}\right]}_{\hat{K}}
\left[\hspace{-2pt}
\begin{array}{c}
\hat{y}\\
\hat{w}
\end{array}
\hspace{-2pt}\right]
\end{equation}
is a stabilizing controller for the new subsystem of interest
\begin{equation}
\left[\hspace{-2pt}
\begin{array}{c}
\hat{y}\\
\hat{w}
\end{array}
\hspace{-2pt}\right]=
\left[\hspace{-2pt}
\begin{array}{c}
G_{yu}^{+}\\
G_{wu}^{+}
\end{array}
\hspace{-2pt}\right]u.
\end{equation}
\end{subequations}
The resultant retrofit controller is specifically found as
\[
\begin{array}{ccl}
u&=&\hat{K}_{y}
\Bigl\{(y-G_{yv}v)\\
&&\hspace{40pt}+G_{yv}\overline{G}_{\rm apx}(I-G_{wv}\overline{G}_{\rm apx})^{-1}(w-G_{wv}v)\Bigr\}\\
&+&\hat{K}_{w}(I-G_{wv}\overline{G}_{\rm apx})^{-1}(w-G_{wv}v).
\end{array}
\]
It is not trivial to see that the control system in Fig.~\ref{figpresys}(b) with such a complicated controller can be equivalently expressed as the cascade block diagram in Fig.~\ref{fignewsigflow}.

The feedback structure in the retrofit controller is encapsulated as the invertible transfer matrix ``$X$" involved in (\ref{stoutretn}), which gives a clear bridge between the new retrofit controller in Theorem~\ref{thmost} and the existing one in Proposition~\ref{propost}.
In fact, those retrofit controllers have a one-to-one correspondence, i.e., $\hat{K}$ is a stabilizing controller for $G_{(y,w)u}^{+}$ in the new retrofit control formulation if and only if $\hat{K}X$ is a stabilizing controller for $G_{(y,w)u}$ in the existing formulation.
We remark that such an idea of factorizing a stabilizing controller for $G_{(y,w)u}$ as in the particular form of ``$\hat{K}X$" is generally difficult to devise in the framework of the existing retrofit control.

Owing to this special controller factorization, the extended retrofit controller gains higher flexibility in design.
In the existing formulation, the Youla parameter of all output-rectifying retrofit controllers is expressed as $Q=\hat{Q}R$ where we can choose $\hat{Q}$ as ``any'' stable transfer matrices.
This means that even the dimension of $\hat{Q}$ can be arbitrary in general.
However, a standard controller design technique, such as the $\mathcal{H}_{2}/\mathcal{H}_{\infty}$-control synthesis, generally produces a stabilizing controller $\hat{K}$, or equivalently $\hat{Q}$, only with a dimension comparable to that of $G_{(y,w)u}$.
This can be seen as an implicit limitation to find a possibly better controller.
In contrast, the new retrofit control formulation provides an additional degree of freedom to find out a higher-dimensional $\hat{Q}$ by tuning the approximate environment model $\overline{G}_{\rm apx}$, whose dimension can be selected arbitrarily.

Another practical insight gained from Theorem~\ref{thmost} is the fact that the gap between the actual performance level of the resultant control system and the assumed performance level of the modular control system can be reduced if accurate environment modeling is performed.
In particular, we can easily have a generalization of Proposition~\ref{propbnd} as follows.

\begin{theorem}\label{thmbnd}\normalfont
With the same notation as that in Theorem~\ref{thmost}, the resultant control performance is bounded as
\begin{equation}\label{Tul2}
|\check{\gamma}^{+}-\hat{\gamma}^{+}|\leq\|T_{zd}\|_{\infty}\leq\hat{\gamma}^{+}+\check{\gamma}^{+}
\end{equation}
where the induced gains of $\hat{z}$ and $\check{z}$ are given as
\[
\hat{\gamma}^{+}\!:=\!\|\hat{T}_{zd}^{+}(\hat{K})\|_{\infty},\quad \!\!
\check{\gamma}^{+}\!:=\!\|G_{zv}(I-\overline{G}G_{wv})^{-1}\Delta\hat{T}_{wd}^{+}(\hat{K})\|_{\infty}.
\]
\end{theorem}

As a generalization of Proposition~\ref{propbnd}, $\hat{\gamma}^{+}$ again corresponds to the assumed performance level, and $\check{\gamma}^{+}$ evaluates the gap between $\|T_{zd}\|_{\infty}$ and $\hat{\gamma}^{+}$.
Because the modeling error $\Delta$ is linearly involved in $\check{\gamma}^{+}$, we can  expect that $\check{\gamma}^{+}$ decreases if the magnitude of $\Delta$ is made small.
Clearly, $\check{\gamma}^{+}=0$, or equivalently, $\|T_{zd}\|_{\infty}=\hat{\gamma}^{+}$ if $\Delta=0$.
Therefore, as improving the accuracy of environment modeling, we can generally reduce the ``bottleneck'' of the existing retrofit control described in Section~\ref{secrevrc}.

\vspace{3pt}
\begin{breakbox}
\noindent\underline{\textbf{NOTE}}~
We again remark that the proposed retrofit control has a clear distinction from robust control.
One may think that the modeling error $\Delta$ can be handled as a model uncertainty in robust control.
However, because the environment $\overline{G}$ is assumed here to be unknown, the norm bound of $\Delta$ is not assurable in the above formulation.
Generally, such an unassured modeling error is not considered in a standard robust control setting.
In contrast, the proposed retrofit control can always ensure the internal stability of the resultant control system, without the assurance of modeling accuracy.
The stability assurance is only reliant on the preexisting system stability as premised in Definition~\ref{defretf}.
Note that the resultant controller is also a retrofit controller, i.e., it can keep the interconnection transfer matrix invariant as shown by (\ref{gengqg}) and (\ref{intinv}).
This property would be counterintuitive because the proposed retrofit controller is designed based on the feedback model of $G$ and $\overline{G}_{\rm apx}$.
\end{breakbox}\vspace{3pt}


\subsection{Time-Domain Analysis: State-Space Realization of Extended Retrofit Controllers}\label{sectimd}

For simplicity of the Youla parameterization, we have assumed in Theorem~\ref{thmost} that the subsystem $G$ is stable, and the approximate environment model $\overline{G}_{\rm apx}$ belongs to $\overline{\mathcal{G}}$.
However, these assumptions are, in fact, not crucial to prove the internal stability of the resultant control system as shown in this subsection.
To prove this, we derive a tractable state-space realization of  $K$ in (\ref{stoutretn}).
Furthermore, we show that the block diagram in Fig.~\ref{fignewsigflow} can be understood as a particular state-space realization of the entire control system obtained by a coordinate transformation.

We describe a time-domain representation of the preexisting system (\ref{prefed}) by
\begin{subequations}\label{timepreu}
\begin{equation}\label{timesys}
\overline{G}:v=\overline{\boldsymbol G}w,\quad
G:\left\{
\begin{array}{ccl}
\dot{x}&=&Ax+Bu+Lv+Wd \vspace{-3pt}\\
w&=&\mathit{\Gamma}x \vspace{-3pt}\\
z&=&Sx \vspace{-3pt}\\
y&=&Cx.
\end{array}
\right.
\end{equation}
For simplicity of description, we suppose here that $\overline{\boldsymbol G}$ is a static map, i.e., a matrix.
We remark that the subsequent discussion can be easily extended to the case of dynamical environments in such a way that $\overline{\boldsymbol G}$ is regarded as the convolution operator associated with $\overline{G}$, i.e., 
\[
 v(t)=\int_{0}^{t}\overline{g}(t-\tau)w(\tau)d\tau
\]
where $\overline{g}(t)$ is the impulse response of $\overline{G}$.
The bold face symbols that will appear in the subsequent discussion, such as $\overline{\boldsymbol G}_{\rm apx}$, are also supposed to be static, just for simplicity of description.

The premise of $\overline{G}\in\overline{\mathcal{G}}$, i.e., the internal stability of (\ref{prefed}), in Section~\ref{secfda} can be rephrased as the stability of 
\begin{equation}\label{timepre}
G_{\rm pre}:
\left\{
\begin{array}{ccl}
\dot{x}&=&(A+L\overline{\boldsymbol G}\mathit{\Gamma})x+Bu+Wd \vspace{-3pt}\\
z&=&Sx \vspace{-3pt}\\
y&=&Cx,
\end{array}
\right.
\end{equation}
which is a combined representation of the subsystem and environment in (\ref{timesys}).
\end{subequations}
As a time-domain analog of Definition~\ref{defretf}, we introduce the following terminology.

\begin{definition}\normalfont
For the preexisting system $G_{\rm pre}$ in (\ref{timepre}), define the set of all admissible environments as
\begin{equation}
\overline{\mathscr G}:=\left\{\overline{\boldsymbol G}:A+L\overline{\boldsymbol G}\mathit{\Gamma}{\rm\ is\ stable}\right\}.
\end{equation}
Under Assumption~\ref{assumvy}, an output feedback controller
\begin{equation}\label{contime}
u=\mathcal{K}(y,w,v),
\end{equation}
where $\mathcal{K}$ denotes a dynamical map, is said to be a \textit{retrofit controller} if the resultant control system that is composed of (\ref{timepreu}) and (\ref{contime}) is internally stable for any possible environment $\overline{\boldsymbol G}\in\overline{\mathscr G}$.
\end{definition}

On the basis of this definition, a state-space realization of the extended retrofit controller is given as follows.
We again remark that Assumptions~\ref{asstaG} and \ref{asmodsta}, i.e., the assumptions on the stability of $G$ and $G^{+}$, are not required to prove the internal stability of the resultant control system.

\begin{theorem}\label{thmtdr}\normalfont
Let Assumption~\ref{assumvy} hold.
For any approximate environment model $\overline{\boldsymbol G}_{\rm apx}$ and any feedback gains $\hat{{\boldsymbol K}}_{y}$ and $\hat{{\boldsymbol K}}_{w}$ such that
\begin{equation}\label{localstab}
A+L\overline{\boldsymbol G}_{\rm apx}\mathit{\Gamma}+B(\hat{{\boldsymbol K}}_{y}C+\hat{{\boldsymbol K}}_{w}\mathit{\Gamma})
\end{equation}
is stable, an output feedback controller
\begin{equation}\label{ssK}
K:\left\{
\begin{array}{ccl}
\dot{\hat{x}}&=&A\hat{x}+L\bigl\{v-\overline{\boldsymbol G}_{\rm apx}(w-\mathit{\Gamma}\hat{x})\bigr\} \vspace{-3pt}\\
u&=&\hat{{\boldsymbol K}}_{y}(y-C\hat{x})+\hat{{\boldsymbol K}}_{w}(w-\mathit{\Gamma}\hat{x})
\end{array}
\right.
\end{equation}
is a retrofit controller.
\end{theorem}

\begin{proof}\normalfont
We first prove that (\ref{ssK}) is a state-space realization of $K$ in (\ref{stoutretn}).
The time-domain representation of $\hat{K}$ in (\ref{stoutretn}) is now given as
\begin{equation}\label{timehatk}
\hat{K}:u=\left[\hspace{-2pt}
\begin{array}{cc}
\hat{{\boldsymbol K}}_{y}&\hat{{\boldsymbol K}}_{w}
\end{array}
\hspace{-2pt}\right]
\left[\hspace{-2pt}
\begin{array}{cc}
\hat{y}\\
\hat{w}
\end{array}
\hspace{-2pt}\right].
\end{equation}
The stability of (\ref{localstab}) corresponds to the fact that the feedback gains $\hat{{\boldsymbol K}}_{y}$ and $\hat{{\boldsymbol K}}_{w}$ are chosen such that (\ref{timehatk}) stabilizes
\begin{equation}\label{upsys}
\overline{G}_{\rm apx}:\hat{v}=\overline{\boldsymbol G}_{\rm apx}\hat{w},\quad
G_{(y,w)u}:\left\{
\begin{array}{ccl}
\dot{\hat{\xi}}&=&A\hat{\xi}+Bu+L\hat{v} \vspace{-3pt}\\
\hat{w}&=&\mathit{\Gamma}\hat{\xi} \vspace{-3pt}\\
\hat{y}&=&C\hat{\xi},
\end{array}
\right.
\end{equation}
which is a time-domain representation of $G_{(y,w)u}^{+}$.
What remains to show is that
\begin{equation}\label{exrecti}
XR:\left\{
\begin{array}{ccl}
\dot{\hat{x}}&=&A\hat{x}+L\bigl\{v-\overline{\boldsymbol G}_{\rm apx}(w-\mathit{\Gamma}\hat{x})\bigr\} \vspace{-3pt}\\
\hat{y}&=&y-C\hat{x}\vspace{-3pt}\\
\hat{w}&=&w-\mathit{\Gamma}\hat{x}
\end{array}
\right.
\end{equation}
is a realization of $XR:(y,w,v)\mapsto(\hat{y},\hat{w})$, i.e., the output rectifier given in the frequency domain as
\[
XR=\left[
\begin{array}{ccc}
\scriptstyle I&\scriptstyle G_{yv}(I-\overline{G}_{\rm apx}G_{wv})^{-1}\overline{G}_{\rm apx}&\scriptstyle-G_{yv}(I-\overline{G}_{\rm apx}G_{wv})^{-1}\\
\scriptstyle 0&\scriptstyle(I-G_{wv}\overline{G}_{\rm apx})^{-1}&\scriptstyle-(I-G_{wv}\overline{G}_{\rm apx})^{-1}G_{wv}
\end{array}
\right].
\]
The block diagram of (\ref{exrecti}) can be depicted as in Fig.~\ref{figexrec}.
From this diagram, we see that
\[
\begin{array}{ccl}
y^{\prime}&=&-G_{yv}(I-\overline{G}_{\rm apx}G_{wv})^{-1}\overline{G}_{\rm apx}w\vspace{-2pt}\\
&&\hspace{100pt}+G_{yv}(I-\overline{G}_{\rm apx}G_{wv})^{-1}v\vspace{1pt}\\
w^{\prime}&=&-(I-G_{wv}\overline{G}_{\rm apx})^{-1}G_{wv}\overline{G}_{\rm apx}w\vspace{-2pt}\\
&&\hspace{100pt}+(I-G_{wv}\overline{G}_{\rm apx})^{-1}G_{wv}v
\end{array}
\]
where we have used the fact that
\begin{equation}\label{relPK2}
(I-PK)^{-1}P=P(I-KP)^{-1}.
\end{equation}
Therefore, we have
\[
\begin{array}{ccl}
\hat{y}&=&y+G_{yv}(I-\overline{G}_{\rm apx}G_{wv})^{-1}\overline{G}_{\rm apx}w\vspace{-2pt}\\
&&\hspace{110pt}-G_{yv}(I-\overline{G}_{\rm apx}G_{wv})^{-1}v\vspace{1pt}\\
\hat{w}&=&(I-G_{wv}\overline{G}_{\rm apx})^{-1}w-(I-G_{wv}\overline{G}_{\rm apx})^{-1}G_{wv}v
\end{array}
\]
where we have used (\ref{relPK1}) for the calculation of $\hat{w}$.
Thus, (\ref{ssK}) is proven to be a state-space realization of $K$ in (\ref{stoutretn}).

In what follows, we show that the resultant control system composed of (\ref{timepreu}) and (\ref{ssK}) is internally stable for any $\overline{\boldsymbol G}\in\overline{\mathscr G}$.
To this end, we apply the coordinate transformation of
\begin{equation}\label{coord}
\hat{\xi}=x-\hat{x},\quad\check{\xi}=\hat{x}.
\end{equation}
Then, we see that the entire control system is given as the cascade connection of the upstream system 
\[
\dot{\hat{\xi}}=\bigl\{
A+L\overline{\boldsymbol G}_{\rm apx}\mathit{\Gamma}+B(\hat{{\boldsymbol K}}_{y}C+\hat{{\boldsymbol K}}_{w}\mathit{\Gamma})
\bigr\}\hat{\xi}
+Wd,
\]
the stability of which is equivalent to that of the closed-loop system of (\ref{timehatk}) and (\ref{upsys}), and the downstream system
\begin{equation}
\dot{\check{\xi}}=(A+L\overline{\boldsymbol G}\mathit{\Gamma})\check{\xi}+L(\overline{\boldsymbol G}-\overline{\boldsymbol G}_{\rm apx})\mathit{\Gamma}\hat{\xi},
\end{equation}
the stability of which is equivalent to $\overline{\boldsymbol G}\in\overline{\mathscr G}$.
Thus, the control system composed of (\ref{timepreu}) and (\ref{ssK}) is equivalent to
\begin{equation}\label{sstzd}
T_{zd}:\left\{
\begin{array}{ccl}
\dot{\hat{\xi}}&=&\bigl\{
A+L\overline{\boldsymbol G}_{\rm apx}\mathit{\Gamma}+B(\hat{{\boldsymbol K}}_{y}C+\hat{{\boldsymbol K}}_{w}\mathit{\Gamma})
\bigr\}\hat{\xi}
+Wd \vspace{-0pt}\\
\dot{\check{\xi}}&=&(A+L\overline{\boldsymbol G}\mathit{\Gamma})\check{\xi}
+L(\overline{\boldsymbol G}-\overline{\boldsymbol G}_{\rm apx})\mathit{\Gamma}\hat{\xi}\vspace{-0pt}\\
z&=&S(\hat{\xi}+\check{\xi}).
\end{array}
\right.
\end{equation}
Hence, the internal stability is proven for any $\overline{\boldsymbol G}\in\overline{\mathscr G}$. \hfill $\square $
\end{proof}

\begin{figure}[t]
\begin{center}
\includegraphics[width=50mm]{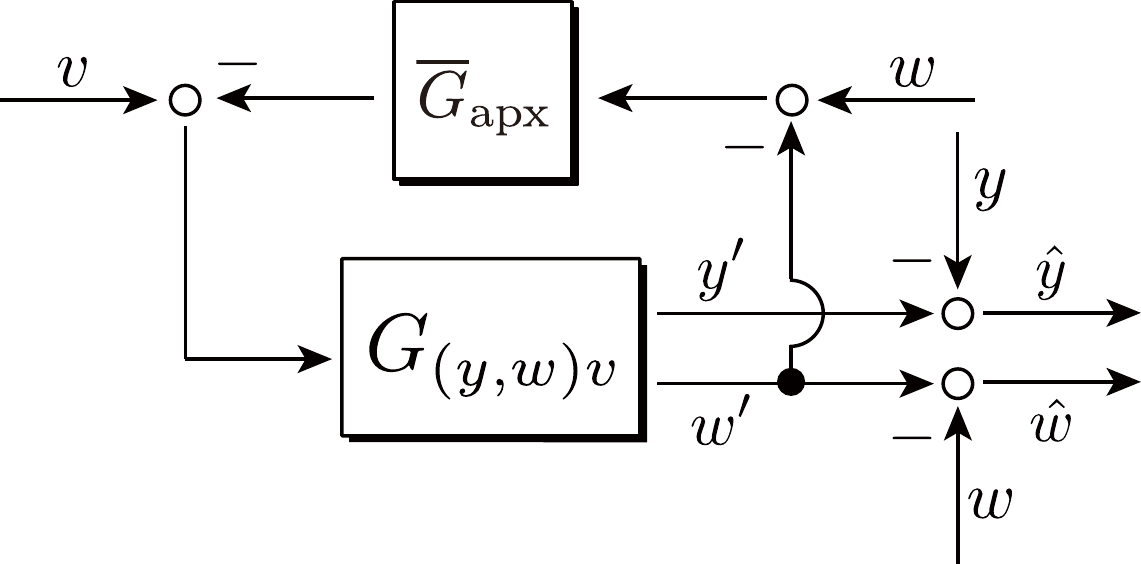}
\end{center}
\vspace{3pt}
\caption{
Block diagram of extended output rectifier.}
\label{figexrec}
\vspace{6pt}
\end{figure}

Theorem~\ref{thmtdr} provides a tractable state-space realization of the retrofit controller in Theorem~\ref{thmost}.
We remark that Theorem~\ref{thmtdr} gives only a ``sufficient" condition to prove that $K$ with the structure of (\ref{ssK}) is a retrofit controller, without imposing Assumptions~\ref{asstaG} and \ref{asmodsta}.
However, it is generally difficult to show by this time-domain analysis if the structure of $K$ in (\ref{ssK}) is ``necessary" or not.
Such necessity of the controller structure is shown by the frequency-domain analysis in Theorem~\ref{thmost}, on the premise of Assumptions~\ref{asstaG} and \ref{asmodsta} making the Youla parameterization tractable.
To prove the necessity of the retrofit controller structure without  Assumptions~\ref{asstaG} and \ref{asmodsta} is left as future work.

\begin{figure}[t]
\begin{center}
\includegraphics[width=65mm]{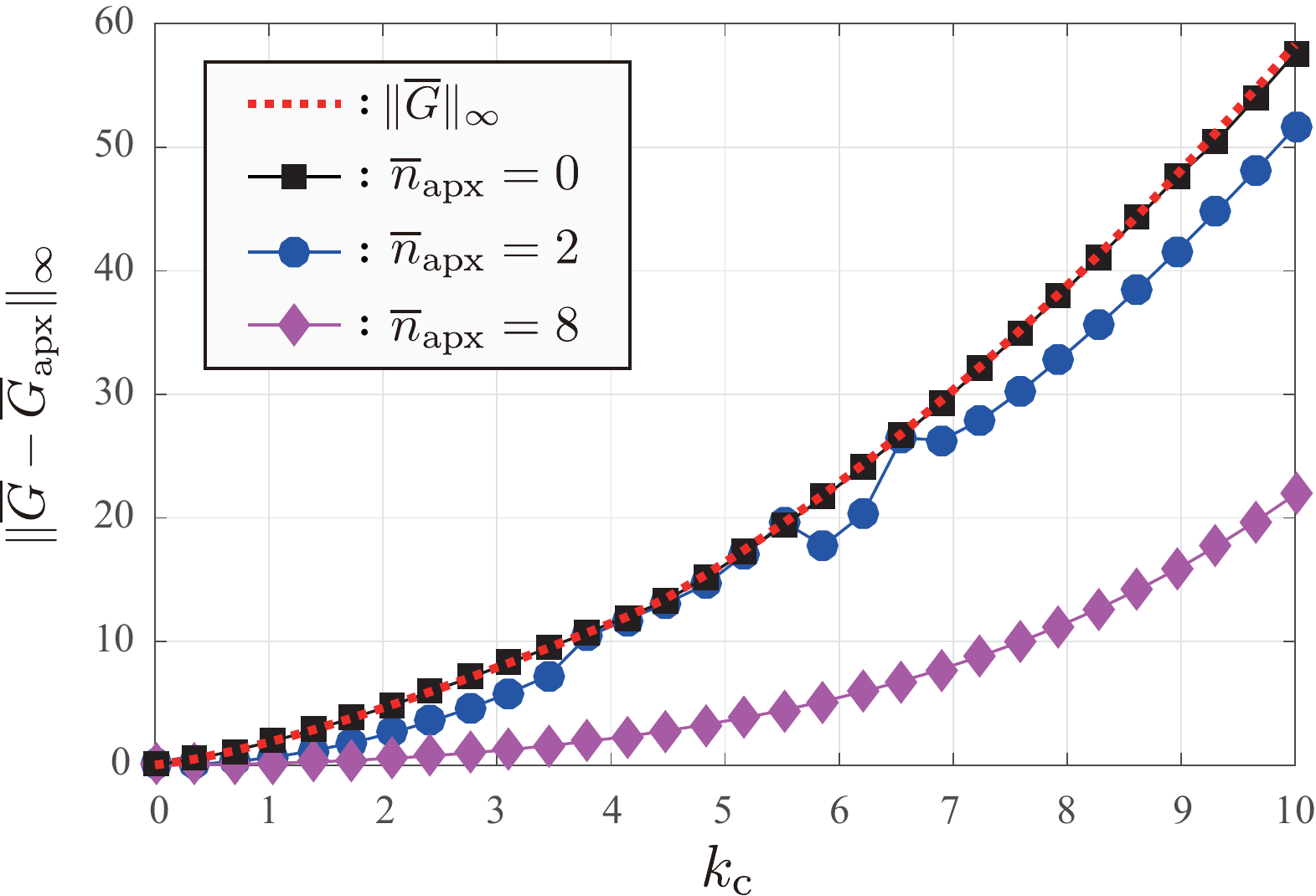}
\end{center}
\vspace{3pt}
\caption{
Modeling errors versus coupling constants between subsystem and environment.}
\label{figerrors}
\vspace{6pt}
\end{figure}

\begin{figure*}[t]
\begin{center}
\includegraphics[width=160mm]{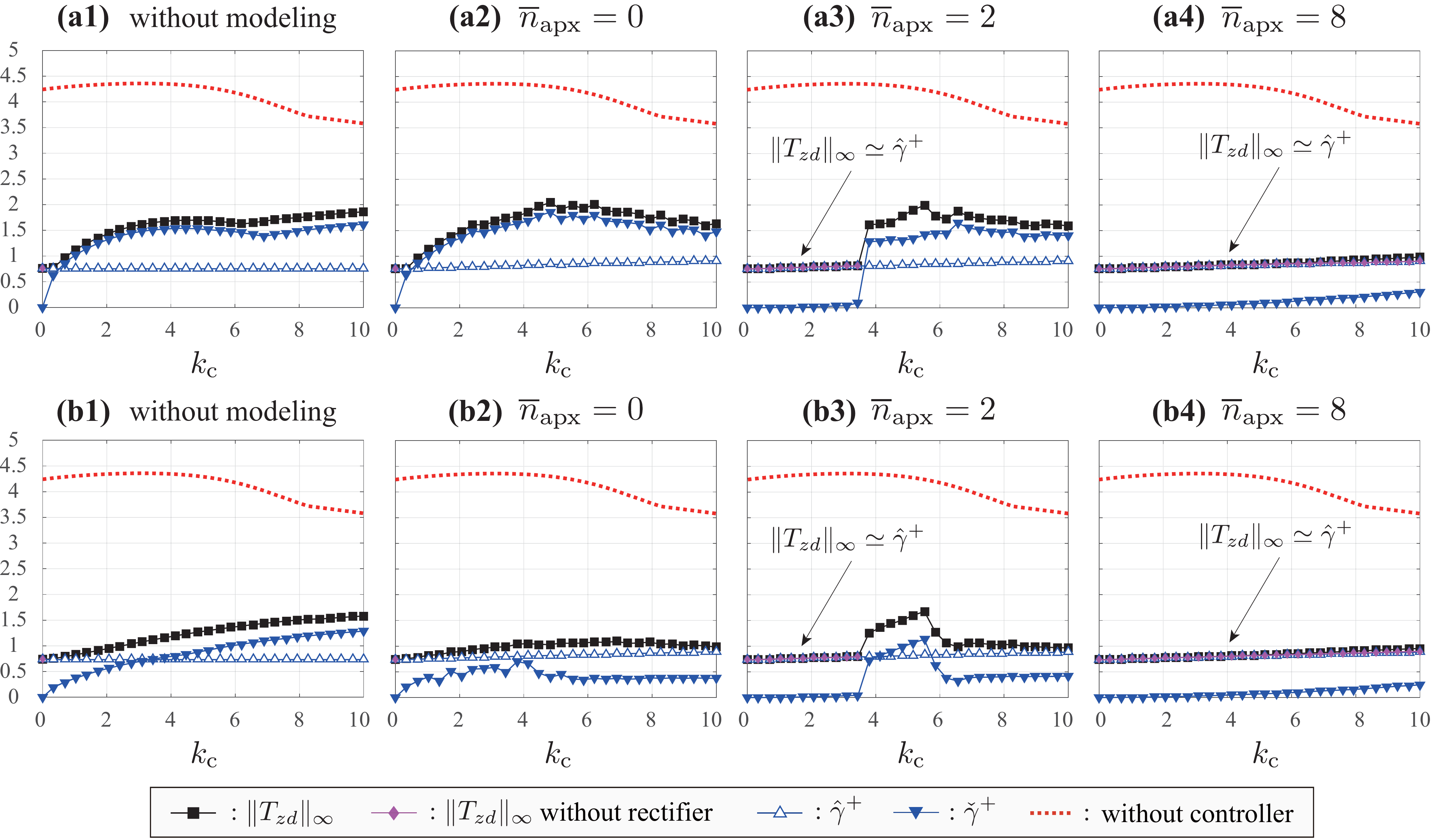}
\end{center}
\vspace{3pt}
\caption{
Resultant control system performance versus coupling strength.
(a1)-(a4) $\alpha=0.2$. (b1)-(b4) $\alpha=0.01$. 
(a1), (b1) Case without environment models, i.e., existing method. 
(a2), (b2) Case with static environment models. 
(a3), (b3) Case with 2-dimensional environment models.
(a4), (b4) Case with 8-dimensional environment models.
}
\label{figperf}
\vspace{6pt}
\end{figure*}

\begin{figure*}[t]
\begin{center}
\includegraphics[width=140mm]{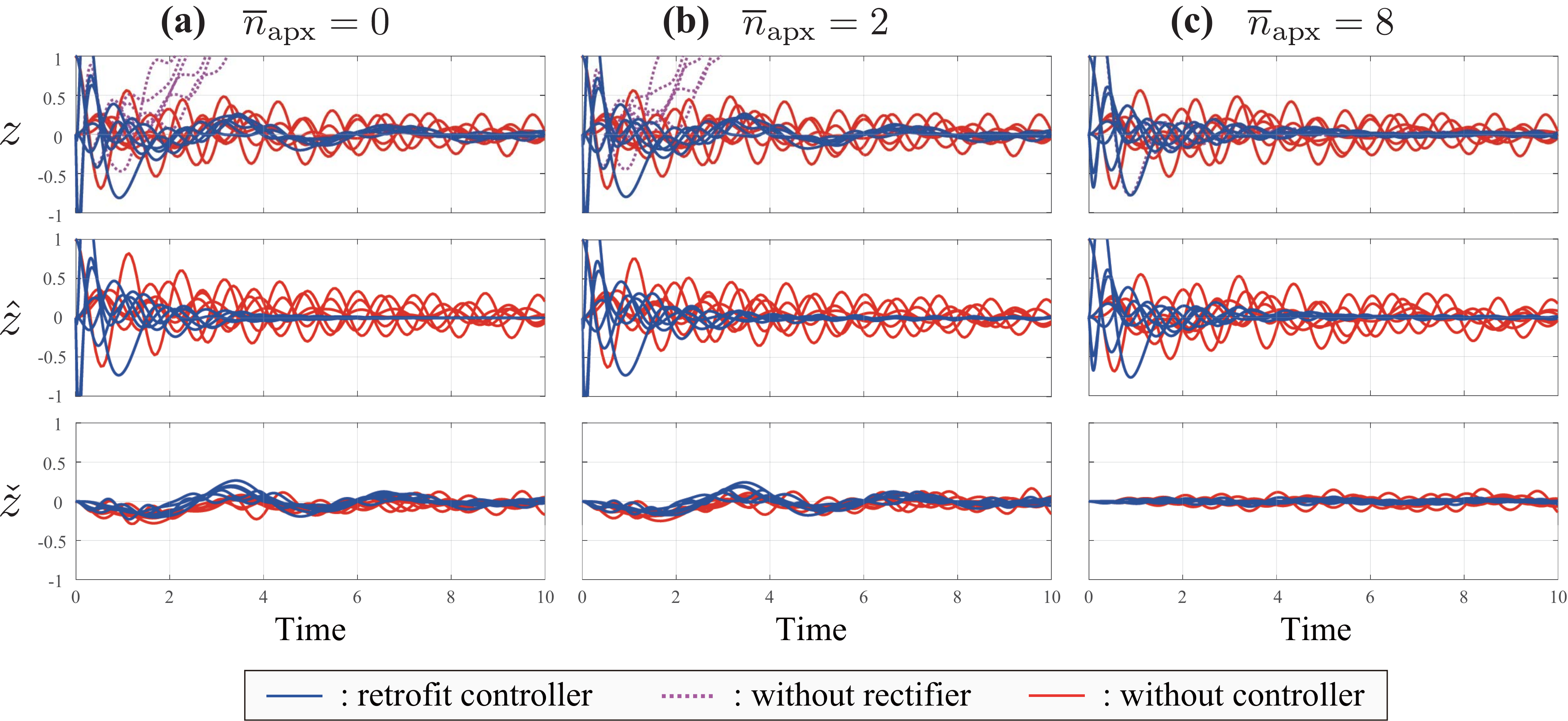}
\end{center}
\vspace{3pt}
\caption{
Resultant control system behavior in response to impulsive disturbance at the first node.
The coupling constant and weighting constant are set as $k_{\rm c}=10$ and $\alpha=0.2$, comparable to Fig.~\ref{figtrajml}(c).
(a) Case with static environment model. 
(b) Case with 2-dimensional environment model.
(c) Case with 8-dimensional environment model.}
\label{figtrajapx}
\vspace{6pt}
\end{figure*}

We can interpret the cascade system (\ref{sstzd}) as a state-space realization of the block diagram in Fig.~\ref{fignewsigflow}.
In particular, we see that
\[
\hat{w}=\mathit{\Gamma}\hat{\xi},\quad\hat{z}=S\hat{\xi},\quad\check{z}=S\check{\xi}.
\]
The transformation between the original realization in Fig.~\ref{figpresys}(b) and the cascade realization in Fig.~\ref{fignewsigflow} is given as the coordinate transformation (\ref{coord}), whose inverse is given as
\[
x=\hat{\xi}+\check{\xi},\quad\hat{x}=\check{\xi}.
\]
The left equation can be seen as decomposition of $x$ into the sum of $\hat{\xi}$ and $\check{\xi}$.
In the previous work \cite{ishizaki2018retrofit}, such state decomposition is discussed in the context of state-space expansion, called hierarchical state-space expansion, where no environment modeling is considered.
From (\ref{sstzd}), we see that $\hat{\xi}$ in the upstream system is directly regulatable by $\hat{K}$ in (\ref{timehatk}), but $\check{\xi}$ in the downstream system is not.
This is compatible with the discussion of the regulatability of $\hat{z}$ and $\check{z}$.

\section{Numerical Experiments}

We again consider the second-order oscillator network in Section~\ref{secmot}.
Just for simplicity of demonstration, we produce approximate environment models applying the balanced truncation \cite{moore1981principal}, which is well known as a standard model reduction technique based on the controllability and observability Gramians.
In practice, the balanced truncation is not directly applicable because a complete model of the environment is supposed to be unknown.
However, it is shown in \cite{kashima2016noise} that such an energy-based model reduction method can be understood as a special case of data-based model reduction methods, e.g., proper orthogonal decomposition (POD), with sufficient data collection of systems driven by stochastic noise.
Based on this fact, we consider simulating a data-rich situation for environment modeling, where we suppose an approximate models produced by the balanced truncation to be an ideally identified model.

Varying the dimension of approximate environment models, denoted by $\overline{n}_{\rm apx}$, we plot in Fig.~\ref{figerrors} the resultant modeling errors $\|\overline{G}-\overline{G}_{\rm apx}\|_{\infty}$ versus the coupling constant $k_{\rm c}$.
From this figure, we see that modeling errors decrease as model dimensions increase, and they increase as coupling constants increase.
The case of $\overline{n}_{\rm apx}=0$ corresponds to the modeling of the static characteristics of $\overline{G}$, which is represented by the static feedback of $\theta_{i}$ in (\ref{intvi}).
We remark that even the 8-dimensional model is not very accurate because it has more than 30\% worst-case error when $k_{\rm c}=10$.
As seen here, an approximate model with moderate dimensions may have a large modeling error.
We remark that the magnitude of modeling errors is not assurable in practice when the environment is unknown and variant.

Varying the dimension of approximate models, we plot the resultant control performance versus the coupling constant in Fig.~\ref{figperf}. 
The black lines with squares represent the actual control performance level $\|T_{zd}\|_{\infty}$, the blue lines with  triangles represent the assumed performance level $\hat{\gamma}^{+}$, the blue lines with inverted triangles represent the performance gap $\check{\gamma}^{+}$, the magenta lines with diamonds correspond to the cases without the output rectifier, and the red dotted lines correspond to the cases without controllers.
Figs.~\ref{figperf}(a1)-(a4) correspond to the cases where the weighting constant $\alpha$ in (\ref{Jalp}) is given as $0.2$, while Figs.~\ref{figperf}(b1)-(b4) correspond to the cases where $\alpha=0.01$.
We remark that the magenta lines with diamonds frequently fall out because the resultant control systems become unstable.
From the set of these plots, we obtain the following observations.
\begin{itemize}
\item[\textbf{--}] The proposed retrofit controller can assure the resultant control system stability for all the cases, while the simple controller not involving the output rectifier induces the system instability for almost all cases where no environment model is used and the static model is used. \vspace{2pt}
\item[\textbf{--}] When the 2-dimensional model is used, both retrofit and simple controllers attain an actual performance level comparable to the assumed performance level provided that the coupling constant is less than about 4. \vspace{2pt}
\item[\textbf{--}] When the 8-dimensional model is used, the retrofit controller and the simple controller attain comparable performance levels for all coupling constants within 0 to 10. \vspace{2pt}
\item[\textbf{--}] The retrofit controller generally attains a performance level better than that in the case where no controller is implemented.
\end{itemize}

As demonstrated here, the stability of resultant control systems is always assured without requiring the assurance of environment modeling accuracy, and the resultant control performance can be improved as improving the accuracy of environment modeling.
These are the major advantages of the proposed retrofit control.
For reference, we plot the impulse response of the resultant control system in Fig.~\ref{figtrajapx}, where Figs.~\ref{figtrajml}(a)-(c) correspond to the cases where the static environment model, the 2-dimensional model, and 8-dimensional model are used, respectively.
We set the parameters as $k_{\rm c}=10$ and $\alpha=0.2$, comparable to Fig.~\ref{figtrajml}(c).
From these results, we can confirm the significance to incorporate approximate environment modeling into retrofit control.

\section{Concluding Remarks}

In this paper, we developed a novel retrofit control method with approximate environment modeling.
We first derived a compact representation of all the proposed retrofit controllers based on a constrained version of the Youla parameterization, where an approximate model of environments is involved as a tuning parameter.
Then, we showed that the resultant control system has a cascade structure, which makes it easy to analyze upper and lower bounds of the resultant control performance. 
Furthermore, we showed that the proposed retrofit controller has a tractable state-space realization with an output rectifier.

The major advantages of the proposed retrofit control are as follows.
\begin{itemize}
\item[\textbf{--}] The stability of the resultant control system is robustly assured regardless of not only the stability of approximate environment models, but also the magnitude of modeling errors, provided that the network system before implementing retrofit control is originally stable.
\item[\textbf{--}] The resultant control performance can be regulated by adjusting the accuracy of approximate environment modeling. 
\end{itemize}
These advantages have a good compatibility with the modeling of unknown environments because the accuracy of identified models may neither be reliable nor assurable in practice, while it can be expected to improve by suitable learning trials.
Such an unassured modeling error is not generally considered in a standard robust control setting.

Incorporating a closed-loop system identification method for online environment modeling would be an interesting direction of future work.
In addition, it is meaningful to conduct a more detailed analysis in the case where multiple retrofit controllers  are simultaneously implemented in respective subsystems based on individual environment modeling.


\bibliographystyle{elsarticle-num}         
\bibliography{IEEEabrv,reference,reference_CREST,reference_aux}            

\begin{thebibliography}{10}
\expandafter\ifx\csname url\endcsname\relax
  \def\url#1{\texttt{#1}}\fi
\expandafter\ifx\csname urlprefix\endcsname\relax\def\urlprefix{URL }\fi
\expandafter\ifx\csname href\endcsname\relax
  \def\href#1#2{#2} \def\path#1{#1}\fi

\bibitem{parnas1972criteria}
D.~L. Parnas, On the criteria to be used in decomposing systems into modules,
  Communications of the ACM 15~(12) (1972) 1053--1058.

\bibitem{huang1998modularity}
C.-C. Huang, A.~Kusiak, Modularity in design of products and systems, IEEE
  Transactions on Systems, Man, and Cybernetics-Part A: Systems and Humans
  28~(1) (1998) 66--77.

\bibitem{schilling2000toward}
M.~A. Schilling, Toward a general modular systems theory and its application to
  interfirm product modularity, Academy of management review 25~(2) (2000)
  312--334.

\bibitem{baldwin2000design}
C.~Y. Baldwin, K.~B. Clark, Design rules: The power of modularity, Vol.~1, MIT
  press, 2000.

\bibitem{ulrich2003role}
K.~T. Ulrich, The role of product architecture in the manufacturing firm,
  Managing in the modular age: architectures, networks, and organizations
  (2003) 117--145.

\bibitem{ulrich2003product}
K.~T. Ulrich, Product design and development, McGraw Hill Higher Education,
  2003.

\bibitem{ishizaki2018retrofit}
T.~Ishizaki, T.~Sadamoto, J.-i. Imura, H.~Sandberg, K.~H. Johansson, Retrofit
  control: Localization of controller design and implementation, Automatica 95
  (2018) 336--346.

\bibitem{sadamoto2018retrofit}
T.~Sadamoto, A.~Chakrabortty, T.~Ishizaki, J.-i. Imura, Retrofit control of
  wind-integrated power systems, IEEE Transactions on Power Systems 33~(3)
  (2018) 2804--2815.

\bibitem{inoue2018parameterization}
M.~Inoue, T.~Ishizaki, M.~Suzumura, Parametrization of all retrofit controllers
  toward open control-systems, in: European Control Conference (ECC), 2018,
  2018, pp. 715--720.

\bibitem{sasahara2018parameterization}
H.~Sasahara, T.~Ishizaki, J.-i. Imura, Parameterization of all state-feedback
  retrofit controllers, in: Decision and Control (CDC), 2018 IEEE 56th Annual
  Conference on, IEEE, 2018, p. to appear.

\bibitem{siljak1972stability}
D.~D. {\v{S}}iljak, Stability of large-scale systems under structural
  perturbations, Systems, Man, and Cybernetics, IEEE Transactions on SMC-2~(5)
  (1972) 657--663.

\bibitem{rotkowitz2006characterization}
M.~Rotkowitz, S.~Lall, A characterization of convex problems in decentralized
  control, Automatic Control, IEEE Transactions on 51~(2) (2006) 274--286.

\bibitem{d2003distributed}
R.~D'Andrea, G.~E. Dullerud, Distributed control design for spatially
  interconnected systems, Automatic Control, IEEE Transactions on 48~(9) (2003)
  1478--1495.

\bibitem{langbort2004distributed}
C.~Langbort, R.~S. Chandra, R.~D'Andrea, Distributed control design for systems
  interconnected over an arbitrary graph, Automatic Control, IEEE Transactions
  on 49~(9) (2004) 1502--1519.

\bibitem{vsiljak2005control}
D.~D. {\v{S}}iljak, A.~I. Ze{\v{c}}evi{\'c}, Control of large-scale systems:
  Beyond decentralized feedback, Annual Reviews in Control 29~(2) (2005)
  169--179.

\bibitem{rantzer2015scalable}
A.~Rantzer, Scalable control of positive systems, European Journal of Control
  24 (2015) 72--80.

\bibitem{ljung1998system}
L.~Ljung, System identification, in: Signal analysis and prediction, Springer,
  1998, pp. 163--173.

\bibitem{bishop2006pattern}
C.~M. Bishop, Pattern recognition and machine learning, Springer, 2006.

\bibitem{langbort2010distributed}
C.~Langbort, J.~Delvenne, Distributed design methods for linear quadratic
  control and their limitations, Automatic Control, IEEE Transactions on 55~(9)
  (2010) 2085--2093.

\bibitem{farokhi2013optimal}
F.~Farokhi, C.~Langbort, K.~H. Johansson, Optimal structured static
  state-feedback control design with limited model information for
  fully-actuated systems, Automatica 49~(2) (2013) 326--337.

\bibitem{petes2017scalable}
R.~Pates, G.~Vinnicombe, Scalable design of heterogeneous networks, Automatic
  Control, IEEE Transactions on 62 (2017) 2318--2333.

\bibitem{stoustrup2009plug}
J.~Stoustrup, Plug \& play control: Control technology towards new challenges,
  European Journal of Control 15~(3-4) (2009) 311--330.

\bibitem{bendtsen2013plug}
J.~Bendtsen, K.~Trangbaek, J.~Stoustrup, Plug-and-play control: Modifying
  control systems online, IEEE Transactions on Control Systems Technology
  21~(1) (2013) 79--93.

\bibitem{willems1972dissipative}
J.~C. Willems, Dissipative dynamical systems part {I}: General theory, part
  {II}: Linear systems with quadratic supply rates, Archive for Rational
  Mechanics and Analysis 45~(5) (1972) 321--393.

\bibitem{moylan1978stability}
D.~Hill, P.~Moylan, Stability criteria for large-scale systems, Automatic
  Control, IEEE Transactions on 23~(2) (1978) 143--149.

\bibitem{qu2014modularized}
Z.~Qu, M.~A. Simaan, Modularized design for cooperative control and
  plug-and-play operation of networked heterogeneous systems, Automatica 50~(9)
  (2014) 2405--2414.

\bibitem{antoulas2005approximation}
A.~C. Antoulas, Approximation of large-scale dynamical systems, Society for
  Industrial and Applied Mathematics, Philadelphia, PA, USA, 2005.

\bibitem{obinata2001model}
G.~Obinata, B.~D. Anderson, Model reduction for control system design, Springer
  London, 2001.

\bibitem{girard2009hierarchical}
A.~Girard, G.~J. Pappas, Hierarchical control system design using approximate
  simulation, Automatica 45~(2) (2009) 566--571.

\bibitem{zhou1995robust}
K.~Zhou, J.~C. Doyle, K.~Glover, Robust and optimal control, Prentice Hall,
  1995.

\bibitem{youla1976modern}
D.~Youla, H.~Jabr, J.~Bongiorno, Modern {Wiener-Hopf} design of optimal
  controllers--{Part II}: The multivariable case, IEEE Transactions on
  Automatic Control 21~(3) (1976) 319--338.

\bibitem{doyle2013feedback}
J.~C. Doyle, B.~A. Francis, A.~R. Tannenbaum, Feedback control theory, Courier
  Corporation, 2013.

\bibitem{falb1967decoupling}
P.~L. Falb, W.~Wolovich, Decoupling in the design and synthesis of
  multivariable control systems, Automatic Control, IEEE Transactions on 12~(6)
  (1967) 651--659.

\bibitem{wang2013distributed}
X.~Wang, N.~Hovakimyan, Distributed control of uncertain networked systems: A
  decoupled design, Automatic Control, IEEE Transactions on 58~(10) (2013)
  2536--2549.

\bibitem{ishizaki2018graph}
T.~Ishizaki, A.~Chakrabortty, J.-i. Imura, Graph-theoretic analysis of power
  systems, Proceedings of the IEEE 106~(5) (2018) 931--952.

\bibitem{moore1981principal}
B.~Moore, Principal component analysis in linear systems: Controllability,
  observability, and model reduction, Automatic Control, IEEE Transactions on
  26~(1) (1981) 17--32.

\bibitem{kashima2016noise}
K.~Kashima, Noise response data reveal novel controllability gramian for
  nonlinear network dynamics., Scientific reports 6 (2016) 27300.

\end{thebibliography}



\end{document}